\documentclass[preprint,1p,times,12pt]{elsarticle} % suggested by Konstantin
\usepackage{lineno}
\usepackage{geometry}
\usepackage{amssymb}
\usepackage{amsmath}
\usepackage{amsfonts}
\usepackage{algorithm}
\usepackage{array}
\usepackage[caption=false,font=normalsize,labelfont=sf,textfont=sf]{subfig}
\usepackage{textcomp}
\usepackage{stfloats}
\usepackage{url}
\usepackage{verbatim}
\usepackage{graphicx}
\usepackage{multirow}
\usepackage{lipsum}
\usepackage{mathtools}
\usepackage{amsthm}
\usepackage{algpseudocode}
\usepackage{arydshln}

\usepackage{booktabs}

\usepackage{soul}
\usepackage{color}
\usepackage{xcolor}

\algrenewcommand\algorithmicrequire{\textbf{Input:}}
\algrenewcommand\algorithmicensure{\textbf{Output:}}

\journal{Signal processing}

\begin{document}

\begin{frontmatter}



\newtheorem{definition}{Definition}[section]
\newtheorem{theorem}[definition]{Theorem}
\newtheorem{proposition}[definition]{Proposition}
\newtheorem{lemma}[definition]{Lemma}
\newtheorem{corollary}[definition]{Corollary}
\newtheorem{remark}[definition]{Remark}
\newtheorem{assumption}[definition]{Assumption}
\newtheorem{conjecture}[definition]{Conjecture}
\newtheorem{example}[definition]{Example}
\newtheorem{property}[definition]{Property}

%\title{Practical polynomial decoupling framework with \textcolor{green!80!black}{Bernstein-Vandermonde regression}}
\title{Tensor-based Multi-layer Decoupling}

 \author[label1]{Joppe De Jonghe}
  \ead{joppe.dejonghe@kuleuven.be}
 \author[label2]{Konstantin Usevich}
     \ead{konstantin.usevich@univ-lorraine.fr}
 \author[label3]{Philippe Dreesen}
     \ead{philippe.dreesen@gmail.com}
 \author[label1]{Mariya Ishteva}
    \ead{mariya.ishteva@kuleuven.be}

 \affiliation[label1]{organization={Dept. of Computer Science, NUMA, KU Leuven},%Department and Organization 
            city={Geel},
            postcode={2440},
            country={Belgium}}

 \affiliation[label2]{organization={CRAN, Universit{\'e} de Lorraine, CNRS},%Department and Organization 
            city={Vandoeuvre-l{\`e}s-Nancy},
            postcode={54500},
            country={France}}

\affiliation[label3]{organization={Dept. Advanced Computing Sciences (DACS), Mathematics Centre Maastricht (MCM), Maastricht University},%Department nd Organization 
            postcode={6229 EN},
            city={Maastricht},
            country={Netherlands}}

%% Abstract
\begin{abstract}

The decoupling of multivariate functions is a powerful modeling paradigm for learning multivariate input-output relations from data. For the single-layer case, established CPD-based methods are available, but the multi-layer case remained largely unexplored. This work introduces a tensor-based framework for multi-layer decoupling, which is based on
ParaTuck-type tensor decompositions and constrained optimization. We provide theoretical justification behind the considered tensor decompositions and parameterizations. Furthermore, we formulate a structured coupled matrix–tensor factorization that incorporates both Jacobian
and function evaluations, together with a bilevel optimization approach for adaptively balancing first- and zeroth-order information. The feasibility of the proposed methodology is illustrated on synthetic systems, a nonlinear system identification benchmark and neural network compression.

\end{abstract}



%% Keywords

\end{frontmatter}

% Different sections

\section{Introduction}

Decoupling multivariate nonlinear functions into compositions of linear transformations and univariate nonlinearities has emerged as a powerful modeling paradigm for learning multivariate nonlinear input-output relations.
The standard single‑layer decoupling form expresses a multivariate map as a linear transformation, followed by a set of univariate nonlinear functions and a final linear transformation~(Figure \ref{fig:single_layer_decoupling}). 
This representation has been successfully applied across a broad range of block-structured modeling problems \cite{dreesen2015decoupling_wiener, hollander2017multivariate, karami2021applying} as well as in neural network compression tasks \cite{zniyed2021tensor}, and is related to the Waring problem for polynomials \cite{iarrobino1999power, landsberg2011tensors} in algebraic geometry.

This work focuses on tensor-based decoupling methods, first introduced by Dreesen et al.~\cite{dreesen2015decoupling}, who showed that stacking Jacobian evaluations of the function under consideration into a third-order tensor admits a canonical polyadic decomposition (CPD). 
This CPD then allows to retrieve the parameters (weights and internal functions) of the decoupled form. 
%This original method relied on the uniqueness of the computed CPD, which can not be guaranteed in noisy or approximate cases that occur in practice. 
%To this end, 
Several works extended the original method with (non-)parametric smoothness constraints for the internal functions \cite{hollander2017multivariate, zniyed2021tensor, decuyper2022decoupling, decuyper2019decoupling} as well as the incorporation of zeroth-order and second-order information to create a coupled matrix-tensor factorization (CMTF) problem \cite{zniyed2021tensor,dreesen2018lvaica}. 
These extended methods are, however, limited to the single-layer case and multi-layer scenarios remain largely unexplored. Yet, multi-layer extensions of the decoupled model can allow to better maintain expressivity under tight parameter budgets compared to a single-layer representation.

%Despite the decade-long use of single‑layer decoupling methods in literature, their extension to deeper, multi-layer structures remains largely unexplored. On the other hand, real‑world nonlinear systems and modern deep neural networks may possess layered interactions between inputs that are better represented by multi-layer decouplings (or a multi-layer representation may simply be more parameter-efficient).

Extending decoupling techniques to multiple layers introduces several challenges. 
First, in the multi-layer setting the Jacobian tensor no longer admits a simple CPD but instead follows a ParaTuck‑$L$ structure. 
Second, while not occurring in the single-layer case, the multi-layer setting introduces additional ambiguities that couple the internal functions across layers, complicating the estimation of the nonlinearities of each layer from a unique ParaTuck-$L$ decomposition alone. Finally, ParaTuck-based estimation algorithms are known to be challenging without incorporating the structure constraints.

%Finally, any theoretical justification is lacking regarding the choice of basis functions, the treatment of bias terms across layers, or the design of optimization algorithms capable of computing multi-layer decoupled representations.

This paper addresses these gaps by introducing a general framework for tensor-based multi-layer decoupling. An earlier conference publication \cite{de2023compressing} focused on an initial algorithm for computing two‑layer decoupled representations and demonstrated it only on the compression of a simple MNIST network. This work significantly extends the conference version by formalizing the notion of an $L$‑layer decoupled representation and showing that the Jacobian tensor of such a representation necessarily admits a ParaTuck‑$L$ decomposition. 
This theoretical connection provides the foundation for recovering the weight matrices and internal functions of the multi-layer decoupling using tensor‑based methods. We also provide theoretical justification regarding the choice of basis functions, treatment of bias terms across layers, and principles of designing optimization algorithms for computing multi-layer decoupled representations.
%To make this recovery feasible, we propose a basis‑function parameterization for the internal nonlinearities that resolves both classical and slice‑wise scaling ambiguities. Building on this parameterization, we derive two practical algorithms for tensor-based multi-layer decoupling: a projected coupled matrix–tensor factorization method (PROJ-CMPT‑$L$) and a coefficient‑based variant (CONSTR‑CMPT-$L$). These algorithms can be seen as generalizations of existing single‑layer approaches and incorporate zeroth‑ and first‑order information from the system under consideration.
We develop two algorithms for computing multi-layer decoupled representations together with an adaptive bilevel strategy for tuning the coupling hyperparameter between Jacobian and function evaluations, which proves essential for achieving high accuracy in practical applications such as neural network compression.

%We validate the proposed framework on a synthetic example, a standard benchmark in nonlinear system identification, and two neural‑network compression case studies (MNIST and FashionMNIST). The results demonstrate that multi-layer decoupling is both feasible and advantageous in cases that can leverage the added flexibility of multiple layers.

The remainder of the paper is structured as follows. 
Section 2 provides a background on the single-layer decoupling problem, Section 3 reviews relevant tensor decompositions, namely CPD and ParaTuck‑$L$. 
Section 4 formalizes the link between Jacobian tensors and multilayer decoupled representations and discusses basis‑function choices and analytic structure. 
Section 5 derives the constrained optimization formulation combining first‑ and zeroth‑order information. 
Section 6 presents the proposed algorithms. 
Section 7 provides numerical experiments on a synthetic example, a standard benchmark in nonlinear system identification and two neural-network compression case studies. Finally, Section 8 concludes the paper.

\subsection*{Notation}
\noindent Vectors, matrices and tensors are denoted with bold lowercase, bold capital and caligraphic letters respectively, e.g., the vector $\mathbf{a}$, matrix $\mathbf{A}$ and tensor $\mathcal{X}$. 
For a matrix $\mathbf{A} \in \mathbb{R}^{I \times J}$, its $i$th row is denoted as $\mathbf{A}_{i,:}$ and $j$th column as $\mathbf{A}_{:,j}$. For a third-order tensor $\mathcal{X} \in \mathbb{R}^{I \times J \times K}$, $\mathcal{X}_{i,:,:}$,  $\mathcal{X}_{:,j,:}$ and $\mathcal{X}_{:,:,k}$ denote its horizontal, lateral and frontal slices \cite{kolda2009tensor}.
The Kronecker delta function is denoted by $\delta_{i,j}$.
The symbols $(.)^\top$, $\otimes$, $\odot$ and $\circ$ represent the transpose operator and the Kronecker, Khatri-Rao and outer products, respectively \cite{kolda2009tensor}.  
The $n$-mode unfolding, $\text{unfold}_{n}(\mathcal{X})$, and vectorization operators, $\text{vec}(\mathcal{X})$ and $\text{vec}(\mathbf{A})$, are as defined by Kolda et al.~\cite{kolda2009tensor}. 
The operator $\lVert . \rVert$ denotes the tensor norm and is defined as the square root of the sum of the squares of its elements. 
Thus, for a matrix $\mathbf{A} \in \mathbb{R}^{I \times J}$ and tensor $\mathcal{X} \in \mathbb{R}^{I \times J \times K}$,
$\lVert\mathbf{A} \rVert = \sqrt{\sum^{I}_{i=1} \sum^{J}_{j=1} a^2_{i,j}}$, 
and $\lVert \mathcal{X} \rVert = \sqrt{\sum^{I}_{i=1} \sum^{J}_{j=1} \sum^{K}_{k=1}  x^2_{i,j,k}}.
$

%\ku{Added text:}
%Note that the decoupling structure corresponds to a feedforward neural network with a single hidden layer with flexible activation functions.
%Note the absence of bias in the activation functions (compared to traditional feed-forward networks), but we assume that the flexibility of the activation function can compensate for the absence of bias term (see Section~\ref{sec:property_analytic}).

%!TEX root = ../elsarticle-template-num.tex
\section{Background on single-layer decoupling}

%\subsection{Single-layer decoupling}
The goal of finding a single-layer decoupled representation, as introduced in  \cite{dreesen2015decoupling}, is to represent a multivariate vector function $\mathbf{f} : \mathbb{R}^m \rightarrow \mathbb{R}^n$ as the product of the input $\mathbf{x} \in \mathbb{R}^m$ and a linear transformation matrix $\mathbf{W}_0 \in \mathbb{R}^{r \times m}$, passed to a layer of univariate nonlinear functions (nonlinearities) $\mathbf{g} = \begin{bmatrix} g^{(1)} & g^{(2)} & \cdots & g^{(r)}\end{bmatrix}^\top$, $g^{(i)}: \mathbb{R} \rightarrow \mathbb{R}$, followed by a final linear transformation matrix $\mathbf{W}_1 \in \mathbb{R}^{n \times r}$ (see Figure \ref{fig:single_layer_decoupling}), so that
\begin{equation}
    \mathbf{f}(\mathbf{x}) = \mathbf{W}_1
    \mathbf{g}(\mathbf{W}_0 \mathbf{x}). \label{eq:single_layer_decoupling}
\end{equation}

\begin{figure}
    \centering
    \includegraphics[width=0.75\textwidth]{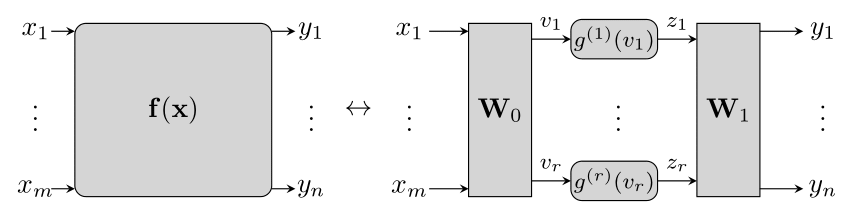}
    \caption{The single-layer decoupling problem. Given a multivariate vector function $\mathbf{f}: \mathbb{R}^{m} \rightarrow \mathbb{R}^{n}$, determine factor matrices $\mathbf{W}_0 \in \mathbb{R}^{r \times m}$, $\mathbf{W}_1 \in \mathbb{R}^{n \times r}$ and the internal functions $\mathbf{g} = \begin{bmatrix}g^{(1)} & g^{(2)} & \cdots & g^{(r)}\end{bmatrix}^\top$, $g^{(i)}: \mathbb{R} \rightarrow \mathbb{R}$, so that
    $\mathbf{f}(\mathbf{x}) = \mathbf{W}_1 \;
    \mathbf{g}(\mathbf{W}_0 \mathbf{x})$.
    }
    \label{fig:single_layer_decoupling}
\end{figure}

%\ku{Added text:}
%Note that the decoupling structure corresponds to a feedforward neural network with a single hidden layer with flexible activation functions.
%Note the absence of bias in the activation functions (compared to traditional feed-forward networks), but we assume that the flexibility of the activation function can compensate for the absence of bias term (see Section~\ref{sec:property_analytic}).

%\subsubsection{Overview of Jacobian-based methods}

To solve the single-layer decoupling problem, one of the key ideas is to use the first-order derivatives of the input-output map $\mathbf{f}(\mathbf{x})$, i.e., its Jacobian $\mathbf{J}_{\mathbf{f}}(\mathbf{x})$. 
Under the assumption of model \eqref{eq:single_layer_decoupling}, the Jacobian admits the factored form
\begin{equation}
    \mathbf{J}_{\mathbf{f}}(\mathbf{x}) = \mathbf{W}_1 \text{diag}(\mathbf{g}'(\mathbf{W}_0\mathbf{x})) \mathbf{W}_0. \nonumber
\end{equation}
The Jacobian matrices, evaluated in different points, can then be stacked into a tensor of which the canonical polyadic decomposition factors are linked to the linear layers in the exact case \cite{dreesen2015decoupling}, and the internal (activation) functions can also be estimated.
The original approach of~\cite{dreesen2015decoupling} proposes a two-step procedure, while the approaches developed in~\cite{hollander2017multivariate} focus on imposing polynomial constraints. 
The work of~\cite{zniyed2021tensor} proposed to combine zeroth- and first-order information, while second-order information is used in~\cite{dreesen2018decoupling}.
A regularization approach to promote smoothness is explored in~\cite{decuyper2019decoupling}, while a recent work in \cite{decuyper2022decoupling} proposes non-parametric filtering for estimation of the internal functions. 
%We should also note that 
The decoupling approach is related to the so-called active subspace technique in machine learning and approximation theory \cite{fornasier2021robust}.
It also has some similarities to score function approaches \cite{janzamin2015beating}, as well as derivative-based methods in independent component analysis \cite{comon2006blind}.
Note that these two approaches deal with single-output maps, and typically need higher-order derivatives (as in the work of~\cite{dreesen2018decoupling}).

\textit{Relation to block-structured and nonlinear system identification} Block structures that result from decoupling are used in nonlinear system identification, as they model complex dynamical systems and have physical interpretation \cite{giri2010block,schoukens2019nonlinear}. 
In this context, several works have applied decoupling methods to nonlinear system identification models such as parallel Wiener-Hammerstein \cite{dreesen2015decoupling_wiener, hollander2016parallel}, polynomial state-space models \cite{dreesen2016decoupling, decuyper2019decoupling}, Volterra series \cite{dreesen2021parameter} and polynomial NARX models \cite{karami2021applying, decuyper2021decoupling}. 
In these works, the mentioned models are used to represent the nonlinear behaviour of several benchmark problems such as the Silverbox \cite{wigren2013three} and the Bouc-Wen model \cite{noel2016hysteretic}. 
After learning the respective model, decoupling methods are applied to either a part of the model, e.g., a nonlinear static block of a Wiener-Hammerstein model, or the model as whole, e.g., for polynomial NARX models. These works highlight the usefulness of decoupling methods in the context of system identification since they provide not only model compression, but also interpretable internal functions that can be related to physical aspects of the underlying system. 

\textit{Relation to neural networks} Traditional neural network training and inference uses fixed activation functions such as ReLU, sigmoid or hyperbolic tangent \cite{apicella2021survey}.
However, decoupled representations can be viewed as neural networks with flexible activation functions that change during the learning of the decoupled network.

Neural networks with flexible activation functions are a timely and relevant topic. 
The parameter efficiency of such networks can be seen as a result of the increased expressivity of neurons with flexible activation functions, over fixed activation functions. 

Several existing works discuss neural networks with flexible activation functions and their approximation and parameter efficiency compared to standard, fixed activation networks. 
For example, Telgarsky~\cite{telgarsky2017neural} discusses neural networks with high-degree rational activation functions and derives error bounds for approximating ReLU networks, Boullé~\cite{boulle2020rational} expands on the work of Telgarsky~\cite{telgarsky2017neural} and Molina~\cite{molinapade}, discussing neural networks with low-degree rational activation functions and proving that these require exponentially fewer parameters than a ReLU network representing the same function. 
Several works discuss the theory and use of splines for learnable activations, some examples include~\cite{yang2024kolmogorov, unser2019representer, balestriero2018spline, parhi2021banach, bohra2020learning}; in particular, Bohra~\cite{bohra2020learning} shows empirically that their introduced deep spline network outperforms the ReLU counterpart and Yang~\cite{yang2024kolmogorov} introduced spline-based Kolmogorov-Arnold networks (KAN) together with promising results on problem cases from the Feynman symbolic regression dataset as well as initial results for MLPs with spline activations, improving image fitting results compared to SIREN~\cite{sitzmann2020implicit}.
This increased expressivity motivates the use of flexible activation functions in the context of decoupling and compression of neural networks.

%!TEX root = ../elsarticle-template-num.tex
\section{Background on tensor decompositions}

\subsection{Canonical polyadic decomposition}

\noindent The canonical polyadic decomposition (CPD) \cite{kolda2009tensor} of a third-order tensor $\mathcal{X} \in \mathbb{R}^{I \times J \times K}$ expresses the tensor as a sum of rank-one terms.
Here we use an equivalent definition~\cite{delathauwer2006linkCPDsimul}: tensor $\mathcal{X}$ admits a CPD if its slices can be written as 
\begin{equation}
    \mathcal{X}_{:,:,k} = \mathbf{A} \cdot \text{diag}(\mathbf{C}_{k,:}) \cdot \mathbf{B}^\top, \text{ for } k=1,2,\hdots,K, \label{eq:CPD}
\end{equation}
where $\mathbf{A} \in \mathbb{R}^{I \times r}$, $\mathbf{B} \in \mathbb{R}^{J \times r}$ and $\mathbf{C} \in \mathbb{R}^{K \times r}$ 
are the three factor matrices.
We use a shorthand notation  $\mathcal{X}  =[\![\mathbf{A}, \mathbf{B}, \mathbf{C}]\!]$.
The tensor rank is defined as the smalles value $r$ for which equation \eqref{eq:CPD} holds. The CPD admits the following scaling and permutation ambiguities \cite{kolda2009tensor}, that is, there exist the following equivalent decompositions of the same tensor::
%These scaling and permutation ambiguities can be described by
\begin{equation}
    \mathcal{X} = [\![\mathbf{A}\mathbf{\Pi}\mathbf{\Lambda_{A}}, \mathbf{B}\mathbf{\Pi}\mathbf{\Lambda_{B}}, \mathbf{C}\mathbf{\Pi}\mathbf{\Lambda_{C}}]\!], \label{eq:ambiguities_CPD}
\end{equation}
with permutation matrix $\mathbf{\Pi} \in \mathbb{R}^{r \times r}$ and diagonal matrices $\mathbf{\Lambda_{A}}$, $\mathbf{\Lambda_{B}}$, $\mathbf{\Lambda_{C}}$ for which $\mathbf{\Lambda_{A}}\mathbf{\Lambda_{B}}\mathbf{\Lambda_{C}}=\mathbf{I}$.
%Several uniqueness conditions exist (for example  Kruskal's condition), see  \cite{sidiropoulos2017tensor} for an overview.

%\subsubsection{Uniqueness of CPD}
%Compared to matrix decompositions, which are typically not unique, tensor decompositions such as the CPD or PT-$2$ have known conditions for uniqueness (up to scaling and permutation ambiguities mentioned in section \ref{sec:ambiguities}).

%For the CPD, the most well-known is Kruskal's condition \cite{kruskal1977three} which provides a sufficient condition for uniqueness and states that a rank-$r$ CPD of a third order tensor $\mathcal{X} = [\![\mathbf{A}, \mathbf{B}, \mathbf{C}]\!]$ is unique if
%\begin{equation}
%    k_{\mathbf{A}} + k_{\mathbf{B}} + k_{\mathbf{C}} \leq 2r + 2. \nonumber
%\end{equation}

\subsection{ParaTuck-$L$ decomposition}

The ParaTuck-$L$ (PT-$L$) decomposition\footnote{Also called ParaTuck-$Z$ \cite{de2019paratuck}, which is a generalization of the ParaTuck-$2$ decomposition, introduced in \cite{harshman1996uniqueness}.} (Figure \ref{fig:PT-k-of-J}) generalizes the CPD for third-order tensors. A tensor $\mathcal{X} \in \mathbb{R}^{I \times J \times K}$ is said to admit a PT-$L$ decomposition if it can be written as
\begin{align}
    \mathcal{X}_{:,:,k} &= \mathbf{W}_L \cdot \mathbf{D}^{(k)}_{L} \cdot \mathbf{W}_{L-1} \cdots \mathbf{D}^{(k)}_{1} \cdot \mathbf{W}_0, \label{eq:PT-k}
\end{align}
for $k=1,2,\hdots,K$, where $\mathbf{W}_L \in \mathbb{R}^{I \times r_L}$, $\mathbf{W}_0 \in \mathbb{R}^{r_1 \times J}$, $\mathbf{W}_{\ell} \in \mathbb{R}^{r_{\ell+1} \times r_{\ell}}$, for $\ell=1,2,\hdots,L-1$ and the $\mathbf{D}^{(k)}_{\ell} \in \mathbb{R}^{r_{\ell} \times r_{\ell}}$, for $\ell=1,2,\hdots,L$, are diagonal matrices. 
We use a shorthand notation
% This work denotes the full PT-$L$ decomposition of $\mathcal{X}$ as
\begin{align}
    \mathcal{X} &= \text{PT}_L(\mathbf{W}_L, \mathbf{W}_{L-1}, \hdots, \mathbf{W}_0,\mathbf{G}^{(L)}, \mathbf{G}^{(L-1)}, \hdots , \mathbf{G}^{(1)}). \nonumber %\\
    %&= \text{PT}_L(\mathbf{W}_{L...0}, \mathbf{G}^{(L...1)}). \nonumber
\end{align}
Here, the rows of the factor matrices $\mathbf{G}^{(\ell)} \in \mathbb{R}^{K \times r_{\ell}}$, for $\ell=1,2,\hdots,L$, correspond to the diagonals of the matrices $\mathbf{D}^{(s)}_{\ell}$, for $\ell=1,2,\hdots,L$, i.e.,
\begin{equation}
\mathbf{D}^{(s)}_{\ell} = \mathrm{diag}( \mathbf{G}^{(\ell)}_{s,:}).
\label{eq:DGcorrespondence}
\end{equation}
The numbers $r_{\ell}$, which are the number of columns of the factor matrices $\mathbf{G}^{(\ell)}$, are called the ParaTuck ranks of the decomposition, denoted as $(r_L, r_{L-1},\hdots,r_1)$.

\begin{figure*}
    \centering
    \includegraphics[width=0.9\textwidth]{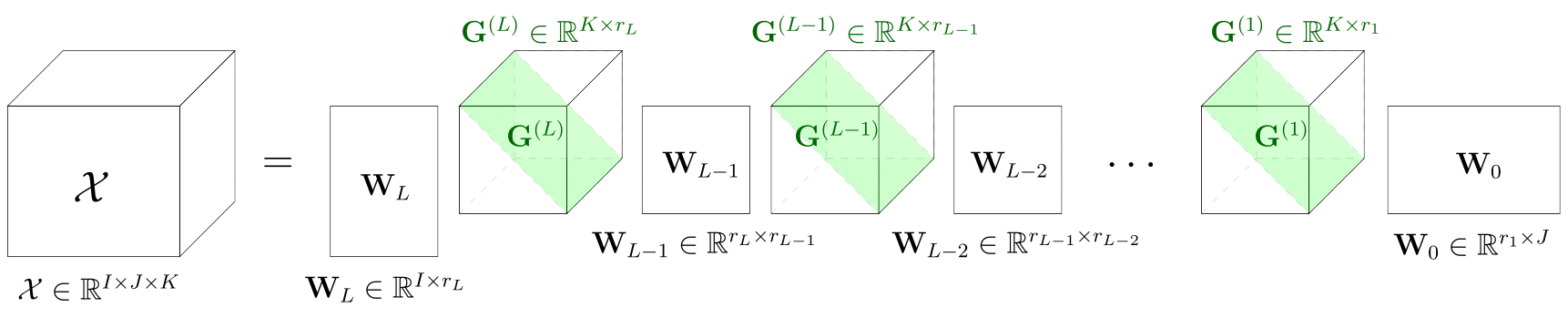}
    \caption{ParaTuck-$L$ decomposition of the tensor $\mathcal{X} \in \mathbb{R}^{I \times J \times K}$, together with the dimensions of the factor matrices $\mathbf{W}_0$, $\mathbf{W}_{\ell}$ and $\mathbf{G}^{(\ell)}$, for $\ell=1,2,\hdots,L$.}
    \label{fig:PT-k-of-J}
\end{figure*}

 For $L=1$, the ParaTuck-$L$ decomposition coincides with the CPD. 
 Thus, 
\begin{equation}
    \mathcal{X} = \text{PT}_1(\mathbf{W}_1, \mathbf{W}_0, \mathbf{G}^{(1)}) = [\![\mathbf{W}_1, \mathbf{W}^\top_0, \mathbf{G}^{(1)}]\!]. \nonumber
\end{equation}

\subsection{Ambiguities of the PT-$L$ decomposition}\label{sec:ambiguities}

 As with the CPD, the ParaTuck-$L$ decomposition has ambiguities because permutation and scaling of the columns/rows of matrices $\mathbf{W}_{\ell}$ result in an alternative decomposition. 
 The following gives a complete description of trivial ambiguities of ParaTuck-$L$: Consider the ParaTuck-$L$ decomposition of a third-order tensor $\mathcal{X} \in \mathbb{R}^{I \times J \times K}$:
\[
\mathcal{X} =  \text{PT}_L({\mathbf{W}}_L, {\mathbf{W}}_{L-1}, \hdots, \mathbf{W}_0,{\mathbf{G}}^{(L)}, \mathbf{G}^{(L-1)}, \hdots , \mathbf{G}^{(1)}).
\]
For $\ell=1,\ldots, L-1$, let $\mathbf{\Pi}_{\ell} \in \mathbb{R}^{r_\ell \times r_\ell}$ be  permutation matrices and  $\mathbf{\Lambda}^{(\ell)}_1,\mathbf{\Lambda}^{(\ell)}_2,\mathbf{\Lambda}^{(\ell)}_3 \in \mathbb{R}^{r_\ell \times r_\ell}$ be diagonal matrices satisfying $ \mathbf{\Lambda}^{(\ell)}_1  \mathbf{\Lambda}^{(\ell)}_2  \mathbf{\Lambda}^{(\ell)}_3 = \mathbf{I}_{r_{\ell}}$;
for convenience, set
\[
\mathbf{\Pi}^{(0)} = \mathbf{\Lambda}^{(0)}_1 = \mathbf{\Lambda}^{(0)}_2 = \mathbf{\Lambda}^{(0)}_3 = \mathbf{I}_{J}, \;
\mathbf{\Pi}^{(L+1)} = \mathbf{\Lambda}^{(L+1)}_1 = \mathbf{\Lambda}^{(L+1)}_2 = \mathbf{\Lambda}^{(L+1)}_3 = \mathbf{I}_{I}.
\]
Let also $\mathbf{\Gamma}^{(\ell)} \in \mathbb{R}^{K \times K}$, for $\ell = 1,\ldots, L-1$ be diagonal matrices satisfying
\[
 \mathbf{\Gamma}^{(1)} \cdots  \mathbf{\Gamma}^{(L-1)}= \mathbf{I}_K,
\]
Then for $\widehat{\mathbf{W}}_{\ell}$ and $\widehat{\mathbf{G}}_{\ell}$, given as
\begin{align*}
\widehat{\mathbf{W}}_{\ell}  =  \mathbf{\Pi}^{\top}_{\ell+1} \mathbf{\Lambda}^{(\ell+1)}_{1} \mathbf{W}_\ell \mathbf{\Lambda}^{(\ell)}_{2} \mathbf{\Pi}_{\ell},\quad
\widehat{\mathbf{G}}_{\ell}  =   \mathbf{\Gamma}^{(\ell)}\mathbf{G}_{\ell} \mathbf{\Lambda}^{(\ell)}_{3} \mathbf{\Pi}_{\ell},
\end{align*}
we have that $\mathcal{X}$ admits the following alternative PT-$L$ decomposition:
\[
\mathcal{X} =  \text{PT}_L(\widehat{\mathbf{W}}_L, \widehat{\mathbf{W}}_{L-1}, \hdots, \widehat{\mathbf{W}}_0,\widehat{\mathbf{G}}^{(L)}, \widehat{\mathbf{G}}^{(L-1)}, \hdots , \widehat{\mathbf{G}}^{(1)}).
\]
%\begin{example}
%Give a small example: for PT-2? CPD? just to explain the notation.
%\end{example}
%\begin{remark}
%In terms of diagonal matrices this gives: (maybe just for the small example).
%\end{remark}
\begin{remark}[On slice-wise ambiguities] The ambiguities involving $\mathbf{\Gamma}^{(\ell)}$ are termed slice-wise scaling ambiguities, as they correspond to simultaneous scaling of the same rows of matrices $\mathbf{G}^{(\ell)}$, that is the diagonal matrices $\mathbf{D}^{k}_{\ell}$ in the $k$-th frontal slice of $\mathcal{X}$.
The CPD does not exhibit slice-wise scaling ambiguities, as can be seen in equation \eqref{eq:ambiguities_CPD}, i.e., $\mathbf{\Gamma}^{(1)} = \mathbf{I}_{K}$. 
As will be discussed later, it is necessary to take these additional ambiguities into account when $L \geq 2$.
\end{remark}

\section{Multilayer decoupling and Jacobian structure}

This section provides a definition of the multi-layer decoupling problem as well as theoretical results on the relation between the ParaTuck-$L$ decomposition and stacked Jacobian matrices of the multi-layer model. In addition, a basis function representation of the internal functions is discussed, together with favorable theoretical properties for the chosen basis.

\subsection{Multi-layer decoupling: formal definition}

The multi-layer decoupling naturally extends the single-layer case as follows.

\begin{definition} \label{def:k-layer-decoupling}
    Let $\mathbf{f}: \mathbb{R}^m \rightarrow \mathbb{R}^n$ be a multivariate vector function. An $L$-layer decoupling of $\mathbf{f}$ is defined as
    \begin{equation}
        \mathbf{f}(\mathbf{x}) = \mathbf{W}_L\mathbf{g}_L(\mathbf{W}_{L-1}\mathbf{g}_{L-1}(\hdots \mathbf{W}_1\mathbf{g}_1(\mathbf{W}_0 \mathbf{x})\hdots)), \label{eq:k-layer-decoupling}
    \end{equation}
    where $\mathbf{W}_{0} \in \mathbb{R}^{r_1 \times m}$, $\mathbf{W}_L \in \mathbb{R}^{n \times r_L}$ and $\mathbf{W}_{\ell} \in \mathbb{R}^{r_{\ell + 1} \times r_{\ell}}$, for $\ell=1,2,\hdots,L-1$. The internal functions $\mathbf{g}_{\ell} : \mathbb{R}^{r_{\ell}} \rightarrow \mathbb{R}^{r_{\ell}}$ are defined as
    \begin{equation}
        \mathbf{g}_{\ell} : \mathbf{u} \mapsto \begin{bmatrix}
                    g^{(1)}_{\ell}(u_1) & g^{(2)}_{\ell}(u_2) & \cdots & g^{(r_{\ell})}_{\ell}(u_{r_{\ell}})\nonumber
                \end{bmatrix}^\top,
    \end{equation}
    where $\mathbf{u} = \begin{bmatrix}
        u_1 & u_2 & \hdots & u_{r_{\ell}}
    \end{bmatrix}^\top \in \mathbb{R}^{{r}_\ell}$ and $g^{(j)}_{\ell}: \mathbb{R} \rightarrow \mathbb{R}$, for $\ell=1,2,\hdots,L$, and  $j=1,2,\hdots,r_{\ell}$. 
    %\mi{this definition is difficult to read}
\end{definition}

\begin{figure*}
    \centering
    \includegraphics[width=0.98\textwidth]{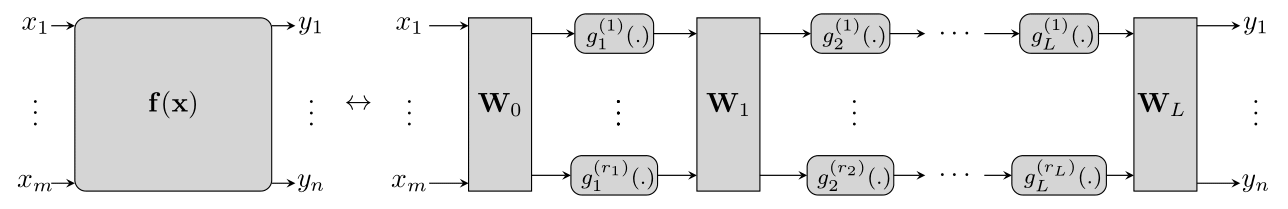}
    \caption{Multi-layer decoupling problem. Given a multivariate vector function $\mathbf{f}$, determine the factor matrices $\mathbf{W}_0$, $\mathbf{W}_1$,$\hdots$, $\mathbf{W}_L$ and the parameters of the internal functions $g^{(j)}_{\ell}$, for $\ell=1,2,\hdots,L$ and $j=1,2,\hdots,r_{\ell}$.}
    \label{fig:multi-layer-decoupling}
\end{figure*}

The $L$-layer decoupling problem (see Figure \ref{fig:multi-layer-decoupling}) is then defined as: given a multivariate vector function $\mathbf{f}: \mathbb{R}^m \rightarrow \mathbb{R}^n$, determine the matrices $\mathbf{W}_0$, $\mathbf{W}_{\ell}$ and internal functions $\mathbf{g}_{\ell}$, for $\ell=1,2,\hdots,L$, such that equation \eqref{eq:k-layer-decoupling} holds. This problem can be viewed in the exact (regardless of uniqueness) or approximate case. 
Here, `exact' means to have an exact representation of $\mathbf{f}(\mathbf{x})$ according to equation \eqref{eq:k-layer-decoupling}, and `approximate' means to have an $L$-layer decoupled representation that approximates the function $\mathbf{f}(\mathbf{x})$ based on some distance measure.

\subsection{Link to ParaTuck-$L$ decomposition}

The described methods use first-order information of the function $\mathbf{f}$, captured by its Jacobian $\mathbf{J_f}$
\begin{equation}
    \mathbf{J_f}(\mathbf{x}) = 
    \begin{bmatrix}
        \dfrac{\partial f_1}{\partial x_1}(\mathbf{x}) & \cdots & \dfrac{\partial f_1}{\partial x_m}(\mathbf{x}) \\
        \vdots & \ddots & \vdots \\
        \dfrac{\partial f_n}{\partial x_1}(\mathbf{x}) & \cdots & \dfrac{\partial f_n}{\partial x_m}(\mathbf{x})
    \end{bmatrix} \in \mathbb{R}^{n \times m}.
    \nonumber
\end{equation}
For an $L$-layer representation, the following Proposition holds.

\begin{proposition}\label{lemma:jac_multilayer}
Given a multivariate vector function $\mathbf{f}: \mathbb{R}^m \rightarrow \mathbb{R}^n$, consider the third order tensor $\mathcal{J} \in \mathbb{R}^{n \times m \times S}$  obtained by stacking the evaluation of $\mathbf{J_f}$ at points   $\mathbf{x}^{(1)}, \mathbf{x}^{(2)},\ldots, \mathbf{x}^{(S)} \in \mathbb{R}^{m}$ as frontal slices
\[
\mathcal{J}_{:,:,s} = \mathbf{J_f}(\mathbf{x}^{(s)}).
\]
If $\mathbf{f}$ admits an $L$-layer decoupling (see Definition \ref{def:k-layer-decoupling}),
then the tensor $\mathbf{\mathcal{J}}$ has the ParaTuck-L decomposition
\[
\mathbf{\mathcal{J}} =  \text{PT}_L(\mathbf{W}_L, \mathbf{W}_{L-1}, \hdots, \mathbf{W}_0,\mathbf{G}^{(L)}, \mathbf{G}^{(L-1)}, \hdots , \mathbf{G}^{(1)}),
\]
where the factor matrices $\mathbf{G}^{(\ell)} \in \mathbb{R}^{S \times r_{\ell}}$ are given as
\begin{equation}
    \mathbf{G}^{(\ell)} = \begin{bmatrix}
     (g^{(1)}_\ell)' (u^{(1)}_{\ell,1}) & \cdots & (g^{(r_\ell)}_\ell)' (u^{(1)}_{\ell,r_{\ell}}) \\
\vdots & & \vdots \\
     (g^{(1)}_\ell)' (u^{(S)}_{\ell,1}) & \cdots & (g^{(r_\ell)}_\ell)' (u^{(S)}_{\ell, r_{\ell}}) \\
    \end{bmatrix}, \label{eq:struct_Dg1}
\end{equation}
where $   (g^{(j)}_\ell)' $ denotes the derivative of $g^{(j)}_\ell$, and $\mathbf{u}^{(s)}_\ell \in \mathbb{R}^{r_{\ell}}$
\begin{align}
\mathbf{u}^{(s)}_\ell =  \begin{bmatrix}u^{(s)}_{\ell,1} & \ldots &  u^{(s)}_{\ell,r_\ell}  \end{bmatrix}^{\top} \label{eq:u_recursive}
\end{align}
are defined recursively as $\mathbf{u}^{(s)}_{1} = \mathbf{W}_0  \mathbf{x}^{(s)}$ and $\mathbf{u}^{(s)}_{\ell} = \mathbf{W}_{\ell-1} \mathbf{g}_{\ell-1}(\mathbf{u}^{(s)}_{\ell-1}), \quad \text{for }\ell = 2,\ldots,L.$
\end{proposition}
\begin{proof}
From the chain rule applied to \eqref{eq:k-layer-decoupling}, we get that 
 \begin{equation}
  \mathcal{J}_{:,:,s}  =   \mathbf{J_f}(\mathbf{x}^{(s)}) = \mathbf{W}_L \cdot \mathbf{D}^{(s)}_{L} \cdot \mathbf{W}_{L-1} \cdots \mathbf{D}^{(s)}_{1} \cdot \mathbf{W}_0, \label{eq:struct_frontal_slices_J}
 \end{equation}
where each diagonal matrix $\mathbf{D}^{(s)}_{\ell} \in \mathbb{R}^{r_{\ell} \times r_{\ell}}$ is given by
\[
  \mathbf{D}^{(s)}_{\ell} =
  \begin{bmatrix}
   (g^{(1)}_\ell)' (u^{(s)}_{\ell,1}) & & \\  & \ddots& \\& & (g^{(r_\ell)}_\ell)' (u^{(s)}_{\ell,r_{\ell}})
   \end{bmatrix}.
\]
Therefore, by definition, the tensor $\mathcal{J}$ admits the ParaTuck-$L$ decomposition \eqref{eq:PT-k},
the factors $\mathbf{G}^{(\ell)}$ of which are given through the correspondence \eqref{eq:DGcorrespondence}.
\end{proof}

\subsection{Parameterizing the internal functions} \label{sec:param_functions}

For designing optimization algorithms, we assume that in each layer the internal functions $\mathbf{g}_{\ell}$ are linear combinations of the basis functions  $\{\phi_{\ell,1},\;\phi_{\ell,2},\hdots,\;\phi_{\ell,d_{\ell}}\}$. 
That is, every internal function can be expressed as
\begin{equation}\label{eq:activation_bases}
    g^{(j)}_{\ell}(u) = c^{(j)}_{\ell,0} + \sum^{d_{\ell}}_{i=1} c^{(j)}_{\ell,i} \cdot \phi_{\ell,i}(u),
\end{equation}
where $d_{\ell}$ is the degree of $g^{(j)}_{\ell}$ and the coefficients $c^{(j)}_{\ell,0}, c^{(j)}_{\ell,1},\ldots,c^{(j)}_{\ell,d_{\ell}}$ are to be learned as part of the decoupling problem \cite{zniyed2021tensor}.
(The proposed algorithms can be generalized to the case of different bases for each internal function).

\subsection{Properties of basis functions: analytic, polynomial and scale-invariance}

\label{sec:property_analytic}
\noindent As shown by the following theorem and corollary, if the internal functions of an $L$-layer decoupled representation of a function $\mathbf{f}(\mathbf{x})$ are polynomial, the bias terms in all but the last layer can be removed. This is a useful property for our algorithms, as will become clear in Section~\ref{sec:algorithm}.
\begin{theorem}\label{theo:bias_theorem}
    Given a multivariate vector function $\mathbf{f}: \Omega \subseteq \mathbb{R}^{m} \rightarrow \mathbb{R}^{n}$ with an exact $L$-layer decoupled representation as in \eqref{eq:k-layer-decoupling}.
    If, for any input $\mathbf{x} \in \Omega$, the transformed input $\mathbf{u}_{\ell} = \mathbf{W}_{\ell}\mathbf{g}_{\ell}(\hdots\mathbf{g}_1(\mathbf{W}_0 \mathbf{x})\hdots)$ for each layer $\ell$ of the decoupled representation in equation \eqref{eq:k-layer-decoupling} belongs to a layer-specific open set $\Omega_{\ell} \subseteq \mathbb{R}^{r_{\ell}}$ such that each internal function $g^{(j)}_{\ell}$, for $j=1,2,\hdots,r_{\ell}$, is real analytic on $\Omega_{\ell}$, then $\mathbf{f}(\mathbf{x})$ admits an equivalent exact decoupled representation
    \begin{equation}
        \mathbf{f}(\mathbf{x}) = \mathbf{W}_L \; \bar{\mathbf{g}}_L(\mathbf{W}_{L-1} \; \hat{\mathbf{g}}_{L-1}(\hdots \hat{\mathbf{g}}_1(\mathbf{W}_0 \mathbf{x})\hdots)), \label{eq:decoupled_analytic_no_constant}
    \end{equation}
    where the constant terms $c^{(j)}_{\ell,0}$ of the internal functions $\hat{g}^{(j)}_{\ell}$ of all but the last layer $\bar{\mathbf{g}}_{L}$, are equal to $0$. Furthermore, all the internal functions are convergent power series.
\end{theorem}
\begin{proof}
The proof is given in \ref{app:proof_theorem}.
\end{proof}

While analytic functions generally admit infinite power‑series representations, restricting attention to polynomials yields the following practically useful corollary.

\begin{corollary}\label{cor:poly_internal}
    Under the assumptions of Theorem \ref{theo:bias_theorem} and supposing furthermore that for each layer $\ell$ the internal functions $g^{(j)}_{\ell}$ are polynomials of degree $d^{(j)}_{\ell}$, then $\mathbf{f}(\mathbf{x})$ admits an equivalent exact decoupled representation
    \begin{equation}
        \mathbf{f}(\mathbf{x}) = \mathbf{W}_L \; \bar{\mathbf{g}}_L(\mathbf{W}_{L-1} \; \hat{\mathbf{g}}_{L-1}(\hdots \hat{\mathbf{g}}_1(\mathbf{W}_0 \mathbf{x})\hdots)), \label{eq:decoupled_analytic_no_constant}
    \end{equation}
    where the constant terms of the internal functions of all but the last layer $\bar{\mathbf{g}}_{L}$, are equal to $0$ and the internal functions $\hat{g}^{(j_1)}_{\ell}$ and $\bar{g}^{(j_2)}_{L}$ remain polynomials of degree $d^{(j_1)}_{\ell}$ and $d^{(j_2)}_{L}$ respectively.
\end{corollary}
The following example illustrates Corollary \ref{cor:poly_internal} for a two-layer system.
\begin{example}
    Consider the two-layer system
    \begin{gather}
        \mathbf{f}(\mathbf{x}) = \mathbf{W}_2\mathbf{g}_2(\mathbf{W}_1\mathbf{g}_1(\mathbf{W}_0\mathbf{x})) \nonumber \\
        \mathbf{W}_2 = \begin{bmatrix}
            1 & 2 \\
            3 & 4
        \end{bmatrix},\; \mathbf{W}_1 = \begin{bmatrix}
            2 & 1 \\
            0 & 1
        \end{bmatrix}, \; \mathbf{W}_0 = \begin{bmatrix}
            1/2 & 1 \\
            2 & 2
        \end{bmatrix} \nonumber \\
        \mathbf{g}_2(\mathbf{v}) = \begin{bmatrix}
            v_1 \\
            - v_2 + v^2_2
        \end{bmatrix}, \mathbf{g}_1(\mathbf{u}) = \begin{bmatrix}
            1 + 3u_1 + u^2_1 \\
            2 + 4u_2 - u^2_2
        \end{bmatrix}. \nonumber
    \end{gather}
    %It is possible to construct an equivalent system where only the internal functions of the second layer contain constant terms different from zero. Namely, 
    Using the construction outlined in the proof of Theorem \ref{theo:bias_theorem} it holds that
    %\begin{equation}
    %    \mathbf{b}_1 = \begin{bmatrix}
    %        1 \\
    %        2
    %    \end{bmatrix} \Rightarrow \hat{\mathbf{b}}_1 = \mathbf{W}_2\mathbf{b}_1 =  \begin{bmatrix}
    %        4 \\
    %        2
    %    \end{bmatrix}.\nonumber
    %\end{equation}
    %As a result, $\hat{\mathbf{g}}(\mathbf{u})$ becomes
    %\begin{equation}
    %    \hat{\mathbf{g}}_1(\mathbf{u}) = \begin{bmatrix}
    %        3u_1 + u^2_1 \\
    %        4u_2 - u^2_2
    %    \end{bmatrix} \nonumber
    %\end{equation}
    %and using equation \eqref{eq:construction_g_2_hat}
    %\begin{align}
    %    \bar{\mathbf{g}}_2(\mathbf{v}) = \begin{bmatrix}
    %        4 + v_1 \\
    %        2 + 3v_2 + v^2_2
    %    \end{bmatrix}. \nonumber
    %\end{align}
    %Thus, the equivalent system is given by
    \begin{gather}
        \mathbf{f}(\mathbf{x}) = \mathbf{W}_2\bar{\mathbf{g}}_2(\mathbf{W}_1\hat{\mathbf{g}}_1(\mathbf{W}_0\mathbf{x})),
        \nonumber \\
%        \mathbf{W}_2 = \begin{bmatrix}
%            1 & 2 \\
%            3 & 4
%        \end{bmatrix},\; \mathbf{W}_1 = \begin{bmatrix}
%            2 & 1 \\
%            0 & 1
%        \end{bmatrix}, \; \mathbf{W}_0 = \begin{bmatrix}
%            1/2 & 1 \\
%            2 & 2
%        \end{bmatrix}, \nonumber \\
        \bar{\mathbf{g}}_2(\mathbf{v}) = \begin{bmatrix}
            4 + v_1 \\
            2 + 3v_2 + v^2_2
        \end{bmatrix}, \hat{\mathbf{g}}_1(\mathbf{u}) = \begin{bmatrix}
            3u_1 + u^2_1 \\
            4u_2 - u^2_2
        \end{bmatrix}, \nonumber
    \end{gather}
    and only the internal functions of the last layer ($\bar{\mathbf{g}}_2$) have non-zero constants.
\end{example}
\noindent

Another useful property that we discuss here deals with the fact that when computing the ParaTuck‑$L$ decomposition, scaling ambiguities arise in the factor matrices $\mathbf{W}_{\ell}$, %for $\ell=0,1,\hdots,L$, 
as discussed in Section \ref{sec:ambiguities}.
In the context of the decoupled model in equation \eqref{eq:k-layer-decoupling}, scaling a row $\mathbf{W}^{i,:}_{\ell}$ by a factor $\alpha$ is equivalent to scaling the input of the corresponding activation function $g_{\ell}^{(i)}$ by $\alpha$. In other words, if the function $\mathbf{f}(x)$ admits an $L$-layer decoupled representation as in Definition \ref{def:k-layer-decoupling}, then the true internal functions are $g^{(i)}_{\ell}(u)$, but in practice the decomposition may yield functions of the form $g^{(i)}_{\ell}(\alpha_{\ell,i}u)$ where the scaling factors $\alpha_{\ell, i}$ are unknown. Because of this ambiguity, it is desirable for the chosen set of basis functions 
% \begin{equation}
$
    \{ \phi_{\ell,1}, \phi_{\ell,2}, \cdots, \phi_{\ell,d_{\ell}}\}$, 
    % \nonumber
% \end{equation}
to possess a property that makes it robust against such scalings. One such property is that the basis should be scale-invariant, meaning that scaling the input does not change the span of the basis functions, as described in the following definition.

\begin{definition}\label{def:scale_invariance}
    Let $\{ \phi_{1}(x), \cdots, \phi_{d}(x)\}$ be a set of basis functions. We call this basis scale-invariant if for any non-zero scalar $\alpha \in \mathbb{R}_0$ the linear span of the basis is preserved under scaling, i.e.,
    \begin{align}
        \text{span}\left(\{ \phi_{1}(x), \cdots, \phi_{d}(x)\}\right) = \text{span}\left(\{ \phi_{1}(\alpha x), \cdots, \phi_{d}(\alpha x)\}\right). \nonumber
    \end{align}
    In other words, scaling the input does not change the function space generated by the basis.
\end{definition}

Finally, ensuring that the chosen basis is scale-invariant according to Definition \ref{def:scale_invariance} guarantees that the expressive power of the assumed decoupled model is unchanged by the scaling ambiguities of the computed ParaTuck-$L$ decomposition.

\subsection{Choice of basis: polynomial and uniqueness}

 \noindent Motivated by the theoretical results of the previous subsection, we choose a polynomial basis, which ensures polynomial internal functions and thereby satisfies the conditions of Corollary \ref{cor:poly_internal}. In addition, the polynomial basis is required to be scale‑invariant according to Definition \ref{def:scale_invariance}. In this work, we focus on a monomial basis
 \begin{equation}
     \{x, x^2, \cdots, x^d\}, \nonumber
 \end{equation}
 which is also the setup considered in \cite{hollander2017multivariate} and \cite{dreesen2015decoupling}. In our study, alternative, `better conditioned' polynomial bases such as a scaled monomial basis or Bernstein basis \cite{lorentz2012bernstein} resulted in comparable experimental results to the monomial basis. Thus, due to simplicity, the focus of this work is on the monomial basis.

With a polynomial representation for the activation functions, the decoupled representation becomes a polynomial neural network \cite{shi2016dpn}.
%Polynomials have the nice property that for any polynomial of degree $d$, an affine transformation of the parameter results also in a polynomial of degree $d$:
%\[
%\widetilde{p}(u) = p(\alpha u + b), \quad
%\]
%for all $\alpha \neq 0$ and $b\in \mathbb{R}$.
%This has several consequences. 
%Rows of matrices $\mathbf{W}_{\ell}$, $\ell = 0,\ldots, L-1$ can be rescaled (multiplied by any non-zero constant) jointly with rescaling the input space of the internal functions, which also leads to (an equivalent) decoupled representation with polynomial bases.
For the single-layer case, the uniqueness of polynomial decoupled representations has been analyzed using tools of algebraic geometry \cite{comon2017identifiability}. In particular, if there are no constant terms ($c^{(j)}_{\ell,0} = 0$), the  polynomial decoupled representation is generically unique as long as $r_1 \le mn$.
The constant terms are not unique if $r_1 > n$, but can be treated separately in practice in the single-layer case \cite{zniyed2021tensor}.

%!TEX root = ../elsarticle-template-num.tex
\section{Practical CMTF framework}\label{sec:practical_CMTF}

While the contributions in previous sections are of theoretical nature, Sections \ref{sec:practical_CMTF} and \ref{sec:algorithm} provide practical algorithms for computing multi-layer decouplings. The proposed decoupling algorithms build on existing Alternating Least Squares (ALS) algorithms, more specifically, the CMTF algorithm introduced by Zniyed \cite{zniyed2021tensor} and the polynomial constraint algorithm introduced by Hollander~\cite{hollander2017multivariate}. Our algorithms take into account the particularities of the multilayer decoupling, use an improved order of updates, and propose a bilevel optimization approach for choosing the regularization parameter. Before introducing the algorithms however, we define in the next subsections the optimization problem to be solved.

\subsection{Solution strategy based on first-order information}

The correspondence between the decoupled representation and the ParaTuck-$L$ decomposition (see Proposition \ref{lemma:jac_multilayer}) suggests the following solution strategy.
\begin{enumerate}
    \item Build the Jacobian tensor $\mathcal{J} \in \mathbb{R}^{n \times m \times S}$ from   Jacobian evaluations in $S$ sampling points $\mathbf{x}^{(1)},\ldots, \mathbf{x}^{(S)}$.
    \item Compute a PT-$L$ decomposition of $\mathcal{J}$ to retrieve the factor matrices $\mathbf{W}_0,\ldots,\mathbf{W}_L$ and $\mathbf{G}^{(1)},\ldots, \mathbf{G}^{(L)}$.
    \item Estimate the internal functions $\mathbf{g}_\ell$ from factor matrices $\mathbf{G}^{(1)}$, $\mathbf{G}^{(2)}$, $\hdots$, $\mathbf{G}^{(L)}$, to determine their representation.
\end{enumerate}

However, there are important issues with the last step of the proposed procedure (estimation of the internal functions).
Unlike the single-layer case, the functions $\mathbf{g}_\ell$ cannot be estimated from the matrix $\mathbf{G}^{(\ell)}$ alone (for example, using the procedure suggested in  \cite{dreesen2015decoupling}). 
This is due to the slice-wise scaling ambiguities mentioned in Section \ref{sec:ambiguities}, which in fact, mix the evaluations of internal functions (see \cite{usevich2023tensor} for an example). Also, there is inherent non-uniqueness of the constant terms in \eqref{eq:activation_bases}. 

This paper tackles these issues by (a) imposing constraints on the factors  $\mathbf{G}^{(\ell)}$ (thanks to the function bases \eqref{eq:activation_bases}), which lead to a constrained ParaTuck-$L$ decomposition, and (b) keeping the nonzero constant terms in \eqref{eq:activation_bases} only for the last layer ($\ell=L$) as suggested by Theorem~\ref{theo:bias_theorem}.
We detail these steps below.

As in \cite{zniyed2021tensor}, we impose the following constraints on the functions. For the matrix $\mathbf{G}$ in \eqref{eq:struct_Dg1} the constraints have the following structure
\begin{align}
        \mathbf{G}^{(\ell)}_{:,j} &= \begin{bmatrix}
            0 & \phi'_{\ell, 1}(u^{(1)}_{\ell, j}) &  \hdots & \phi'_{\ell, d_{\ell}}(u^{(1)}_{\ell, j}) \\
            0 & \phi'_{\ell, 1}(u^{(2)}_{\ell, j}) & \hdots & \phi'_{\ell, d_{\ell}}(u^{(2)}_{\ell, j}) \\
            \vdots & \vdots & & \vdots \\
            0 & \phi'_{\ell, 1}(u^{(S)}_{\ell, j}) & \hdots & \phi'_{\ell, d_{\ell}}(u^{(S)}_{\ell, j}) \\
        \end{bmatrix}
        \cdot
        \begin{bmatrix}
            c^{(j)}_{\ell,0} \\
            c^{(j)}_{\ell,1} \\
            \vdots \\
            c^{(j)}_{\ell,d_{\ell}}
        \end{bmatrix}
        = \mathbf{X}^j_{\ell} \cdot \mathbf{c}^j_{\ell}, \label{eq:cnstr_col_D} 
\end{align}
where $\mathbf{X}^j_{\ell} \in \mathbb{R}^{S \times (d_{\ell} + 1)}$, $\mathbf{c}^j_{\ell} \in \mathbb{R}^{(d_{\ell} + 1) \times 1}$, for $\ell=1,2,\hdots,L$ and $\ell$, $j=1,2,\hdots,r_{\ell}$ and where the $\mathbf{u}^{(s)}_{\ell}$ are defined as in \eqref{eq:u_recursive}.

Note here that the constraints on the columns of the $\mathbf{G}^{(\ell)}$ matrices work across the frontal slices of the tensor $\mathcal{J}$. Because of this, the used method of enforcing the structure on the $\mathbf{G}^{(\ell)}$ matrices indirectly solves the slice-wise scaling ambiguities problem. More specifically, it is still possible that the $\mathbf{G}^{(\ell)}$ matrices are scaled by a factor $\alpha_{\ell}$, but this scaling factor is now the same over all frontal slices of $\mathcal{J}$, which is not the case without the constraints.

\subsection{Combining first- and zeroth-order information}

Inspired by \cite{zniyed2021tensor}, we combine first- and zeroth-order information of the system $\mathbf{f}(\mathbf{x})$ under consideration. The first-order information corresponds to the PT-$L$ decomposition of the tensor $\mathcal{J}$, as described by Proposition \ref{lemma:jac_multilayer}. For the zeroth-order information, consider the following matrix $\mathbf{F} \in \mathbb{R}^{n \times S}$
\begin{align}
        \mathbf{F} &= \begin{bmatrix}
            \mathbf{f}(\mathbf{x}^{(1)}) & \mathbf{f}(\mathbf{x}^{(2)}) & \cdots & \mathbf{f}(\mathbf{x}^{(S)})
        \end{bmatrix} \nonumber \\
        &= \mathbf{W}_L
        \cdot
        \begin{bmatrix}
            g^{(1)}_L(u^{(1)}_{L,1}) & \hdots & g^{(1)}_L(u^{(S)}_{L,1}) \\
            g^{(2)}_L(u^{(1)}_{L,2}) & \hdots & g^{(2)}_L(u^{(S)}_{L,2}) \\
            \vdots &  & \vdots \\
            g^{(r_L)}_L(u^{(1)}_{L,r_L}) & \hdots & g^{(r_L)}_L(u^{(S)}_{L,r_L})
        \end{bmatrix} \label{eq:matrix_F} \\
        &= \mathbf{W}_L \cdot \begin{bmatrix}
            \mathbf{g}_L(\mathbf{u}_{L}^{(1)}) & \mathbf{g}_L(\mathbf{u}_{L}^{(2)}) & \cdots & \mathbf{g}_L(\mathbf{u}_{L}^{(S)})
        \end{bmatrix} \nonumber \\
        &= \mathbf{W}_L \cdot \mathbf{R}^\top, \nonumber
\end{align}
where the $\mathbf{u}^{(s)}_{L}$ are defined as in \eqref{eq:u_recursive}.
 
The factor matrix $\mathbf{W}_L$
is the same factor matrix as for the PT-$L$ decomposition of the tensor $\mathcal{J}$. Our idea, explained below, is to  combine the PT-$L$ decomposition of $\mathcal{J}$ with the factorization of $\mathbf{F}$, which yields a \textit{coupled matrix-tensor factorization} (CMTF) \cite{liu2021tensors}.

Similarly to \eqref{eq:cnstr_col_D}, we can impose structure on $\mathbf{R}$.
If the internal functions are represented as in equation \eqref{eq:activation_bases} , then $\mathbf{R}$ can be expressed as
\begin{align}
        \mathbf{R}_{:,j} &= \begin{bmatrix}
            1 & \phi_{L, 1}(u^{(1)}_{L, j}) & \hdots & \phi_{L, d_{L}}(u^{(1)}_{L, j}) \\
            1 & \phi_{L,1}(u^{(2)}_{L, j}) & \hdots & \phi_{L, d_{L}}(u^{(2)}_{L, j}) \\
            \vdots & \vdots & & \vdots \\
            1 & \phi_{L,1}(u^{(S)}_{L, j}) & \hdots & \phi_{L, d_{L}}(u^{(S)}_{L, j}) \\
        \end{bmatrix}
        \cdot
        \begin{bmatrix}
            c^{(j)}_{L,0} \\
            c^{(j)}_{L,1} \\
            \vdots \\
            c^{(j)}_{L,d_{L}}
        \end{bmatrix}
        = \mathbf{Y}^j_L \cdot \mathbf{c}^j_{L}. \label{eq:cnstr_col_R}
\end{align}

Combining the tensor decomposition of $\mathcal{J}$ with the factorization of matrix $\mathbf{F}$ and adding the constraint on the columns of $\mathbf{R}$ yields the final CMTF optimization problem
\begin{align}
        \displaystyle{
        \min_{\mathclap{\substack{\{\mathbf{W}_{\ell}\}^{L}_{\ell=0},\\\{\mathbf{G}^{(\ell)}\}^{L}_{\ell=1}, \mathbf{R}}}}
        } 
        \;\;\;\;\;\;  & \lVert \mathcal{J} - \text{PT}_L(\mathbf{W}_{L},\ldots,\mathbf{W}_{0},\mathbf{G}^{(L)},\ldots,\mathbf{G}^{(1)}) \rVert^2 
        + \lambda \cdot \lVert \mathbf{F} - \mathbf{W}_L \cdot \mathbf{R}^\top \rVert ^2  \nonumber \\%[-1.5em] %&\quad\quad\quad\quad\quad\quad\quad\quad\quad\quad\quad\quad+ \lambda \cdot \lVert \mathbf{F} - \mathbf{W}_L \cdot \mathbf{R}^\top \rVert^2 \nonumber \\
        \operatorname{s.t.} \quad &\mathbf{G}^{(1)}_{:,j_1} = \mathbf{X}^{j_1}_1 \cdot \mathbf{c}^{j_1}_{1}, \;\;j_1 = 1,2,\hdots,r_1, \nonumber \\
        &\mathbf{G}^{(2)}_{:,j_2} = \mathbf{X}^{j_2}_2 \cdot \mathbf{c}^{j_2}_{2}, \;\;j_2 = 1,2,\hdots,r_2, \nonumber \\
        & \;\;\;\;\;\;\;\;\; \vdots \nonumber \\
        &\mathbf{G}^{(L)}_{:,j_L} = \mathbf{X}^{j_L}_L \cdot \mathbf{c}^{j_L}_{L}, \;\; j_L = 1,2,\hdots,r_L, \nonumber \\
        &\mathbf{R}_{:,j} = \mathbf{Y}^{j}_L \cdot \mathbf{c}^{j}_{L}, \;\;\;\;\; j = 1,2,\hdots,r_L.
        \label{opt:cmtf_PTk}
\end{align}
Similarly to \cite{zniyed2021tensor}, the parameter $\lambda$ determines how much weight is given to the matrix factorization.
In Section \ref{sec:choosing_lambda} we discuss how to choose the parameter $\lambda$ adaptively and we propose a strategy that allows to take into account evaluation metrics for the task under consideration.
 It is important to note here that the constraint on the columns of $\mathbf{R}$ allows to estimate the constant terms of the internal functions only in the last layer of the decoupling. Thus, only the last layer is responsible for correcting the bias and the constant terms of the internal functions in the other layers remain equal to their initial value.

\noindent By applying the tensor decomposition methods on tensors obtained from first-order information of a system's input-output functions, we obtain a solution strategy for `decoupling' the given system into structures having flexible activation functions. 
This work extends the tensor method solution strategies of \cite{dreesen2015decoupling}, \cite{zniyed2021tensor} and \cite{de2023compressing} from the single and two-layer case to the multi-layer case.

\section{Algorithm}\label{sec:algorithm}

%In this section we introduce $2$ multi-layer decoupling algorithms, PROJ-CMPT-$L$ and CONSTR-CMPT-$L$, along with a bilevel optimization scheme for selecting an appropriate $\lambda$. Both rely on alternating optimization but differ in their approach of satisfying the constraints in optimization problem \eqref{opt:cmtf_PTk}. PROJ-CMPT-$L$ does so by projecting the columns of factor matrices $\mathbf{G}^{(i)}$ and $\mathbf{R}$ onto the column spaces of $\mathbf{X}_i$ and $\mathbf{Y}_L$ respectively, while CONSTR-CMPT-$L$ enforces the constraints by directly updating the coefficient vectors $\mathbf{c}^j_i$.

In this section we introduce two multi-layer decoupling algorithms, PROJ-CMPT-$L$ and CONSTR-CMPT-$L$, along with a bilevel optimization scheme for selecting an appropriate $\lambda$. Both algorithms rely on alternating optimization and share several factor matrix updates, but differ in their approach of satisfying the constraints in optimization problem \eqref{opt:cmtf_PTk}: PROJ-CMPT-$L$ via projections of $\mathbf{G}^{(i)}$ and $\mathbf{R}$, CONSTR-CMPT-$L$ via direct updates of the coefficient vectors $\mathbf{c}^j_i$. The general alternating optimization scheme is introduced first and subsequent subsections give more details about the PROJ-CMPT-$L$, CONSTR-CMPT-$L$ and bilevel optimization respectively.

\subsection{Alternating minimization algorithm}\label{sec:CMPT-L}

\noindent To solve the CMTF optimization problem \eqref{opt:cmtf_PTk}  we choose an alternating minimization strategy.
At each step of the algorithm, factors are updated successively (for example update $\mathbf{W}_L$ while keeping $\mathbf{W}_{L-1},\ldots, \mathbf{W}_{0}, \mathbf{G}^{(L)},\ldots, \mathbf{G}^{(1)}$ fixed), as in the standard alternating least-squares method. 
%As mentioned in section \ref{sec:practical_CMTF} we introduce two algorithms in this work, the general structure of these algorithms however is the same and they differ only in how they handle the updating of the coefficients $\mathbf{c}_{\ell}$ and matrices $\mathbf{G}^{(\ell)}$ and $\mathbf{R}$, for $\ell=1,\hdots,L$. 

\begin{algorithm}[t!]
\caption{CMPT-$L$, alternating algorithm structure}\label{alg:CMPT-L}
    \begin{algorithmic}[1]
        \Require $\mathcal{J}$, $\mathbf{F}$, $rArr = \{r_\ell\}^{L}_{\ell = 1}$, $dArr = \{d_{\ell}\}^{L}_{\ell = 1}$, $\phi Arr = \{\{\{\phi^{(j)}_{\ell, d'}\}^{d_{\ell}}_{d' = 1}\}^{r_\ell}_{j = 1}\}^{L}_{\ell = 1}$, $samples$, $\lambda$
        \State Initialize  $\mathbf{R}, \mathbf{W}_0,\ldots,\mathbf{W}_{L},\mathbf{G}^{(1)}, \ldots, \mathbf{G}^{(L)}$
        \While{stopping criterion not met}
            \State $\mathbf{W}_0 = \text{arg}\displaystyle{\min_{\mathbf{W}_0}} \;\lVert  \text{unfold}_2(\mathcal{J})^\top - \mathbf{M}^{(0)}_{\mathbf{W}} \cdot \mathbf{W}_0 \rVert^2$
            \\
            \hrulefill
            \For{$\ell=1,2, \hdots, L-1$}
                \State $\mathbf{G}^{(\ell)}, \mathbf{c}_{\ell}$ $\leftarrow$ Update\_$\mathbf{c}_\ell$($\mathcal{J}$, $\{\mathbf{W}_{\ell}\}^L_{\ell=0}$, $\{\mathbf{G}^{(\ell)}\}^L_{\ell=1}$, $\ell$)
                \State $\mathbf{W}_{\ell}  = \text{arg}\displaystyle{\min_{\mathbf{W}_{\ell}}} \lVert \text{vec}(\mathcal{J}) - \mathbf{M}^{(\ell)}_{\mathbf{W}} \cdot \text{vec}(\mathbf{W}_{\ell}) \rVert^2$
            \EndFor
            \\
            \hrulefill
            \State $\mathbf{G}^{(L)}, \mathbf{R}, \mathbf{c}_{L}$ $\leftarrow$ Update\_$\mathbf{c}_L$($\mathcal{J}$, $\mathbf{F}$, $\{\mathbf{W}_{\ell}\}^L_{\ell=0}$, $\{\mathbf{G}^{(\ell)}\}^L_{\ell=1}$, $\lambda$)
            \State $\mathbf{W}_L  = \text{arg}\displaystyle{\min_{\mathbf{W}_L}} \;\lVert  \text{unfold}_1(\mathcal{J}) - \mathbf{W}_L \cdot \mathbf{M}^{(L)}_{\mathbf{W}} \rVert^2 + \lambda \lVert \mathbf{F} - \mathbf{W}_L \cdot \mathbf{R}^{\top} \rVert^2$
        \EndWhile
        \Ensure $\{\mathbf{W}_{\ell}\}^{L}_{\ell = 0}$, $\{\mathbf{G}^{(\ell)}\}^{L}_{\ell = 1}$, $\mathbf{R}$, $\{\mathbf{c}_{\ell}\}^{L}_{\ell = 1}$
    \end{algorithmic}
\end{algorithm}

Algorithm \ref{alg:CMPT-L} gives the general structure of our introduced Coupled Matrix ParaTuck-$L$ (CMPT-$L$) algorithm. The algorithm uses two types of updates: (a) updating the factors $\mathbf{W}_{\ell}$ while ignoring the constraints, and (b) updating $\mathbf{G}^{(\ell)}$ and $\mathbf{c}_{\ell}$ (and $\mathbf{R}$ in case of last layer).
For the factor matrix updates (a), consider updating the factor $\mathbf{W}_{0}$.
We can show that the cost function, for fixed  $\mathbf{W}_{L},\ldots, \mathbf{W}_{1}, \mathbf{G}^{(L)},\ldots, \mathbf{G}^{(1)}$, can be expressed in a standard least-squares form as
\begin{align*}
 \| \mathcal{J} - \text{PT}_L(\mathbf{W}_{L},\ldots,\mathbf{W}_{0},\mathbf{G}^{(L)},\ldots,\mathbf{G}^{(1)}) \|^2
= 
\lVert  \text{unfold}_2(\mathcal{J})^\top - \mathbf{M}^{(0)}_{\mathbf{W}} \cdot \mathbf{W}_0 \rVert^2,
\end{align*}
for some matrix $\mathbf{M}^{(0)}_{\mathbf{W}}$.
Similarly, we can define the matrices, $\mathbf{M}^{(\ell_1)}_{\mathbf{W}}$ and $\mathbf{M}^{(\ell_2,s)}_{\mathbf{G}}$, for $\ell_1=0,1,\hdots,L$ and $\ell_2=1,2\hdots,L$, whose structure can be found in \ref{app:mat_proj}. Here, there is an important difference with the standard alternating least squares: the constraints in  \eqref{opt:cmtf_PTk} are nonlinear, because the matrices $\mathbf{X}^{j}_\ell$ and $\mathbf{Y}^{j}_L$ depend on the previous factors $\mathbf{W}_{\ell-1}, \ldots, \mathbf{W}_{0}$. We propose to relax this constraint by ignoring these nonlinear dependencies when updating the factor $\mathbf{W}_{\ell}$ (and keep the factors $\mathbf{W}_L, \ldots,\mathbf{W}_{\ell+1}$ and  $\mathbf{G}^{(L)}, \ldots,\mathbf{G}^{(\ell+1)}$ fixed).
%Second, we propose the following strategy for updating $\mathbf{G}^{(\ell)}$: we update $\mathbf{G}^{(\ell)}$ ignoring all the constraints, and then project on the set of contraints.
%\item CMTF-CONSTR-PT-$L$ algorithm (section \ref{sec:CMTF-CONSTR-PT-k}): we update $\mathbf{G}^{(\ell)}$ by solving a least-squares problem over $\mathbf{c}_\ell^j$.
%\end{itemize}
%This strategy is similar to the algorithm in  \cite{zniyed2021tensor} for the single-layer case.
%and  \cite{hollander2017multivariate} in the single-layer case.

In the multilayer case, the order of the updates plays an important role. 
Algorithm \ref{alg:CMPT-L} updates the factors in the following order:
\[
\mathbf{W}_{0},\mathbf{G}^{(1)},\mathbf{W}_{1},\ldots,\mathbf{G}^{(L)},\mathbf{W}_{L},
\]
while, e.g., the order in \cite{zniyed2021tensor} is $\mathbf{W}_{1},\mathbf{W}_{0},\mathbf{G}^{(1)}$  (for $L=1$).

In Algorithm \ref{alg:CMPT-L}, the input $rArr$ is a list specifying the number of internal functions in each layer of the decoupled representation, these values correspond to the ParaTuck ranks of the PT-$L$ decomposition to be computed. The list $dArr$ specifies the degrees of the internal functions in each layer and $\phi Arr$ defines the sets of basis functions used to represent the internal functions in each of the $L$ layers.

Finally, in Algorithm~\ref{alg:CMPT-L} the updates of $\mathbf{G}^{(\ell)}$, $\mathbf{c}_{\ell}$ and $\mathbf{R}$ are left implicit through the functions Update\_$\mathbf{c}_{\ell}$(.) and Update\_$\mathbf{c}_{L}$(.). 
Different strategies can be used to update $\mathbf{G}^{(\ell)}$, $\mathbf{c}_{\ell}$ and $\mathbf{R}$, yielding different algorithms. 
The following subsections introduce two strategies, which as discussed at the start of Section~\ref{sec:practical_CMTF}, yields algorithms that can be seen as multi-layer generalizations of those introduced in \cite{zniyed2021learning} and \cite{hollander2017multivariate}.

\subsection{PROJ-CMPT-$L$ algorithm}\label{sec:CMTF-PT-k}
\noindent The PROJ-CMPT-$L$ algorithm updates $\mathbf{G}^{(\ell)}$ and $\mathbf{c}_{\ell}$ by first updating $\mathbf{G}^{(\ell)}$ in a least square sense, ignoring the constraints, and then projecting $\mathbf{G}^{(\ell)}$ onto the set of constraints. For the case when $\ell = L$, the matrix $\mathbf{R}$ is included in the updates.

Algorithms \ref{alg:c_ell_proj} and \ref{alg:c_L_proj} show the Update\_$\mathbf{c}_{\ell}$(.) and Update\_$\mathbf{c}_{L}$(.) functions for the PROJ-CMPT-$L$ algorithm. The least-squares update and projection steps are separated by a horizontal line.
%For projection of the factors on the set of constraints, Algorithm~\ref{alg:CMTF-PT-K_algorithm} uses a projection strategy, shown  in Algorithms~\ref{alg:projection_strat_l}--\ref{alg:projection_strat_L}. 
In these algorithms, $\mathbf{X}_{\ell}$, for $\ell=1,2,\hdots,L$ and $\mathbf{Y}_L$ are block-diagonal matrices with diagonal elements $\mathbf{X}^{1}_{\ell}$, $\mathbf{X}^{2}_{\ell}$,$\hdots$, $\mathbf{X}^{r_{\ell}}_{\ell}$ and $\mathbf{Y}^{1}_L$, $\mathbf{Y}^{2}_L$,$\hdots$, $\mathbf{Y}^{r_L}_L$, constructed according to  \eqref{eq:cnstr_col_D} and \eqref{eq:cnstr_col_R}.

\begin{algorithm}
\caption{Update\_$\mathbf{c}_{\ell}$ for PROJ-CMPT-$L$}\label{alg:c_ell_proj}
    \begin{algorithmic}[1]
        \Require $\mathcal{J}$, $\{\mathbf{W}_{\ell}\}^L_{\ell=0}$, $\{\mathbf{G}^{(\ell)}\}^L_{\ell=1}$, $\ell$
            \State \parbox[t]{\linewidth}{$\mathbf{G}^{(\ell)}_{s,:} = \text{arg}\displaystyle{\min_{\mathbf{G}^{(\ell)}_{s,:}}} \lVert \text{vec}(\mathcal{J}_{:,:,s}) - \mathbf{M}^{(\ell,s)}_{\mathbf{G}} \ (\mathbf{G}^{(\ell)}_{s,:})^\top \rVert^2$, for $s=1,2,\hdots,S$}
            \\ \hrulefill
            \State Compute $\mathbf{X}^j_{\ell}$, for $j=1,2,\hdots,r_{\ell}$
            \State \parbox[t]{\linewidth}{$\mathbf{c}_{\ell} = \text{arg}\displaystyle{\min_{\mathbf{c}_{\ell}}} \lVert \text{vec}(\mathbf{G}^{(\ell)}) - \mathbf{X}_{\ell} \cdot \mathbf{c}_{\ell} \rVert^2$}
            \State Compute $\mathbf{G}^{(\ell)}_{:,j} = \mathbf{X}^j_{\ell} \cdot \mathbf{c}^{j}_{\ell}$, for $j=1,2,\hdots,r_{\ell}$
        \Ensure $\mathbf{G}^{(\ell)}$, $\mathbf{c}_{\ell}$
    \end{algorithmic}
\end{algorithm}

\begin{algorithm}
\caption{Update\_$\mathbf{c}_{L}$ for PROJ-CMPT-$L$}\label{alg:c_L_proj}
    \begin{algorithmic}[1]
        \Require $\mathcal{J}$, $\mathbf{F}$, $\{\mathbf{W}_{\ell}\}^L_{\ell=0}$, $\{\mathbf{G}^{(\ell)}\}^L_{\ell=1}$, $\mathbf{R}$, $\lambda$
            \State \parbox[t]{\linewidth}{$\mathbf{G}^{(L)}_{s,:}  = \text{arg}\displaystyle{\min_{\mathbf{G}^{(L)}_{s,:}}} \lVert \text{vec}(\mathcal{J}_{:,:,s}) - \mathbf{M}^{(L,s)}_{\mathbf{G}} \cdot (\mathbf{G}^{(L)}_{s,:})^\top\rVert^2$, for $s=1,2,\hdots,S$}
           \State $\mathbf{R} = \text{arg}\displaystyle{\min_{\mathbf{R}}} \lVert \mathbf{F} - \mathbf{W}_L \cdot \mathbf{R}^{\top} \rVert^2$
            \\ \hrulefill
            \State Compute $\mathbf{X}^j_{L}$,$\mathbf{Y}^j_L$, for $j=1,2,\hdots,r_L$
            \State \parbox[t]{\linewidth}{$\mathbf{c}_{L} = \text{arg}\displaystyle{\min_{\mathbf{c}_{L}}} \lVert \text{vec}(\mathbf{G}^{(L)}) \text{-} \mathbf{X}_L \mathbf{c}_{L} \rVert^2 +  \lambda \lVert \text{vec}(\mathbf{R}) \text{-} \mathbf{Y}_L  \mathbf{c}_{L} \rVert^2$}
            \State Compute $\mathbf{G}^{(L)}_{:,j} = \mathbf{X}^j_{L}  \mathbf{c}^{j}_{L}$,  $\mathbf{R}^{:,j} = \mathbf{Y}^j_L \mathbf{c}^j_{L}$,  $j=1,\ldots,r_L$
        \Ensure $\mathbf{G}^{(L)}$, $\mathbf{R}$, $\mathbf{c}_{L}$
    \end{algorithmic}
\end{algorithm}

\subsection{CONSTR-CMPT-$L$ algorithm}\label{sec:CMTF-CONSTR-PT-k}

\noindent In the CONSTR-CMPT-$L$ algorithm, we update the matrix $\mathbf{G}^{(\ell)}$ in such a way that it already satisfies the constraints (similarly to the approach of \cite{hollander2017multivariate} for the single-layer case).
More precisely, we assume that the matrix $\mathbf{G}^{(\ell)}$ satisfies the constraints in \eqref{opt:cmtf_PTk} and minimize the function over $\mathbf{c}_{\ell}$ instead.
For example, the Jacobian-related term can be rewritten as 
\[
\| \mathcal{J} - \text{PT}_L(\mathbf{W}_{L},\ldots,\mathbf{W}_{0},\mathbf{G}^{(L)},\ldots,\mathbf{G}^{(1)}) \|^2  =
\|\text{vec}(\mathcal{J}) - (\mathbf{M}^{(\ell)}_{\mathbf{C}})_0 \cdot \mathbf{c}_{\ell}\|^2
\]
where the matrix $(\mathbf{M}^{(\ell)}_{\mathbf{C}})_0 $ is provided in \ref{app:mat_constr}.

In the case of $\mathbf{c}_{L}$, the matrix factorization of $\mathbf{F}$ also needs to be taken into account, adding to the update formula
\begin{align}
       \text{vec}(\mathbf{F}^\top) = \text{vec}(\mathbf{R} \cdot \mathbf{W}^\top_L) 
       = (\mathbf{W}_L \otimes \mathbf{I}_L) \cdot \text{vec}(\mathbf{R}) = (\mathbf{W}_L \otimes \mathbf{I}_L) \cdot \mathbf{Y}_L \cdot \mathbf{c}_{L}, \label{eq:struct_incorp_update_R_cgk}
\end{align}
where the matrix $\mathbf{Y}_L \in \mathbb{R}^{r_L S \times r_L(d_{L} + 1)}$ is constructed as in Algorithm~\ref{alg:c_L_proj}.

Algorithms~\ref{alg:c_ell_constr} and \ref{alg:c_L_constr} show the Update\_$\mathbf{c}_{\ell}$(.) and Update\_$\mathbf{c}_{L}$(.) functions for the CONSTR-CMPT-$L$ algorithm. Compared to the PROJ-CMPT-$L$ algorithm, there is no projection strategy and the updates for $\mathbf{G}^{(\ell)}$ (and $\mathbf{R}$) are replaced by directly updating the coefficients $\mathbf{c}_{\ell}$. Similarly to the PROJ-CMPT-$L$ algorithm, $\mathbf{Y}_L$ is a block-diagonal matrix with diagonal elements $\mathbf{Y}^{1}_L, \mathbf{Y}^{2}_L, \hdots, \mathbf{Y}^{r_L}_L$ (analogous for $\mathbf{X}_{\ell})$.

\begin{algorithm}
\caption{Update\_$\mathbf{c}_{\ell}$ for CONSTR-CMPT-$L$}\label{alg:c_ell_constr}
    \begin{algorithmic}[1]
        \Require $\mathcal{J}$, $\{\mathbf{W}_{\ell}\}^L_{\ell=0}$, $\{\mathbf{G}^{(\ell)}\}^L_{\ell=1}$, $\ell$
            \State Compute $\mathbf{X}^j_{\ell}$, for $j=1,2,\hdots,r_{\ell}$
            \State $\mathbf{c}_{\ell} = \text{arg}\displaystyle{\min_{\mathbf{c}_{\ell}}} \lVert \text{vec}(\mathcal{J}) - (\mathbf{M}^{(\ell)}_{\mathbf{C}})_0 \cdot \mathbf{c}_{\ell} \rVert^2$
            \State Compute $\mathbf{G}^{(\ell)}_{:,j} = \mathbf{X}^j_{\ell} \cdot \mathbf{c}^{j}_{\ell}$, for $j=1,2,\hdots,r_{\ell}$
        \Ensure $\mathbf{G}^{(\ell)}$, $\mathbf{c}_{\ell}$
    \end{algorithmic}
\end{algorithm}

\begin{algorithm}
\caption{Update\_$\mathbf{c}_{L}$ for CONSTR-CMPT-$L$}\label{alg:c_L_constr}
    \begin{algorithmic}[1]
        \Require $\mathcal{J}$, $\mathbf{F}$, $\{\mathbf{W}_{\ell}\}^L_{\ell=0}$, $\{\mathbf{G}^{(\ell)}\}^L_{\ell=1}$, $\lambda$
            \State Compute $\mathbf{X}^j_{L}$,$\mathbf{Y}^j_L$, for $j=1,2,\hdots,r_L$
            \State \parbox[t]{\linewidth}{
            {
                $
               \displaystyle
               \mathbf{c}_{L} = \text{arg}\displaystyle{\min_{\mathbf{c}_{L}}} \;\; \lVert \text{vec}(\mathcal{J}) - (\mathbf{M}^{(L)}_{\mathbf{C}})_0 \cdot \mathbf{c}_{L} \rVert^2 + \lambda \cdot \lVert \text{vec}(\mathbf{F}^\top) - (\mathbf{W}_L \otimes \mathbf{I}_{S}) \cdot \mathbf{Y}_L \cdot \mathbf{c}_{L} \rVert^2 \nonumber
               $
            }
            }
            \State Compute $\mathbf{G}^{(L)}_{:,j} = \mathbf{X}^j_{L}  \mathbf{c}^{j}_{L}$,  $\mathbf{R}^{:,j} = \mathbf{Y}^j_L \mathbf{c}^j_{L}$,  $j=1,\ldots,r_L$
        \Ensure $\mathbf{G}^{(L)}$, $\mathbf{R}$, $\mathbf{c}_{L}$
    \end{algorithmic}
\end{algorithm}

\subsection{Stopping criteria and choosing $\lambda$}\label{sec:choosing_lambda}
\noindent
In practical scenarios, the PROJ-CMPT-$L$ and CONSTR-CMPT-$L$ algorithm need a stopping criterion and a value for $\lambda$.
The optimal $\lambda$ will typically be unknown, but %needs to be neither small nor large to decrease the loss of accuracy while preventing overfitting.
should strike a balance between maintaining accuracy and preventing overfitting.
We propose to choose $\lambda$  adaptively, starting from a small value (that gives more weight to the Jacobian term), and increasing it over the iterations, as detailed below.

%It is favorable to start from a low value for $\lambda$ and increase it over time while making sure the objective function is still decreasing. Ideally, this aspect of $\lambda$ needs to be taken into account by the chosen stopping criterion.

We propose a two-stage (bilevel) optimization approach (see Algorithm~\ref{alg:practical_use_algorithm}, inputs $rArr$, $dArr$ and $\phi Arr$ are as mentioned in Section \ref{sec:CMPT-L}).
At the first (inner) stage, the value of $\lambda$ is fixed, and the chosen decoupling algorithm, denoted $get\_decoupling(..)$, is run until the inner stopping criterion is satisfied.
To account for possible non-monotonicity of the cost function value in decoupling iterations, the following inner stopping criterion is used: stop if the objective function value has not decreased for a fixed number of iterations or if a maximum number of iterations is reached.

At the second (outer) stage, the value of $\lambda$ is increased while $stop\_metric(t)$ $\leq$ $stop\_metric(t - 1)$, which checks the
performance of the computed decoupling based on some metric. 
The metric can be different from the cost function, can be adapted to a particular task, and can be more expensive to compute than the cost function.
For example, the accuracy on the validation set can serve as a good metric for the network compression task, as will be shown in the next section. 

\begin{algorithm}
\caption{Two-stage decoupling with adaptive $\lambda$}\label{alg:practical_use_algorithm}
\begin{algorithmic}[1]
        \Require $\mathcal{J}$, $\mathbf{F}$, $rArr = \{r_\ell\}^{L}_{\ell = 1}$, $dArr = \{d_{\ell}\}^{L}_{\ell = 1}$, $\phi Arr = \{\{\{\phi^{(j)}_{\ell, d'}\}^{d_{\ell}}_{d' = 1}\}^{r_\ell}_{j = 1}\}^{L}_{\ell = 1}$, $samples$
        \State $t \gets 0$
        \State $\lambda \gets \text{initialize to small value}$
        \State $\beta \gets \text{initialize to value } > 1$
        \State $stop\_metric(t-1) \gets \text{large value}$
        \State $stop\_metric(t) \gets stop\_metric(t-1) + 1$
        \State $Params \gets (\mathcal{J},\;\mathbf{F},\;rArr, \;dArr,\;\phi Arr, \;samples)$ 
        \While{$stop\_metric(t) \leq stop\_metric(t-1)$}
            \State $decoupling(t+1) \gets get\_decoupling(Params, \lambda)$
            \State $\lambda \gets \lambda \cdot \beta$
            \State $t \gets t + 1$
            \State $stop\_metric(t) \gets comp\_metric(decoupling(t))$
        \EndWhile
        \Ensure $decoupling(t-1)$
    \end{algorithmic}
\end{algorithm}

\section{Experiments}
\noindent In this section, we validate our algorithms on synthetic examples, the Silverbox nonlinear system identification benchmark~\cite{wigren2013three} and neural networks trained on MNIST~\cite{deng2012mnist} and FashionMNIST data~\cite{xiao2017fashion}.

\subsection{General setup}
\noindent For the experiments, the CMTF decoupling algorithm introduced by Zniyed et al.~\cite{zniyed2021tensor} is referred to as \textit{Zniyed} and the polynomial constraint decoupling algorithm introduced by Hollander \cite{hollander2017multivariate} is referred to as \textit{Hollander}. Since the Hollander algorithm uses only first-order information, an additional bias correction is needed, we apply the bias $\mathbf{b} \in \mathbb{R}^n$,
\begin{equation}
    \mathbf{b} = \dfrac{1}{S}\mathbf{1}^{\top}_S \left(\mathbf{F} - \mathbf{W}_1 \mathbf{R}^{\top}\right)^{\top}. \nonumber
\end{equation}
The algorithms are prefixed with \textit{TS-} if bi-level optimization (Algorithm~\ref{alg:practical_use_algorithm}) is used, so \textit{TS-Zniyed}, \textit{TS-PROJ-CMPT-$L$} and \textit{TS-CONSTR-CMPT-$L$} refer to the bi-level execution for the respective decoupling algorithm. If multiple configurations of the \textit{TS-PROJ-CMPT-}$L$ or \textit{TS-CONSTR-CMPT-}$L$ algorithms are used in an experiment, e.g., \textit{TS-PROJ-CMPT-}$1$ and \textit{TS-PROJ-CMPT-}$2$, then \textit{TS-PROJ-...} or \textit{TS-CONSTR-...} are used to refer to these configurations in general.

The bi-level algorithms are initialized with $\lambda=1e-6$ and $\beta = 100$. 
The entries of the factor matrices are initialized with random values between $0.1$ and $10$. 
The used $stop\_metric$, number of layers $L$, PT-ranks $r_{\ell}$ and internal degrees $d_{\ell}$ are discussed in the respective experiments. 
The relative error on the tensor $\mathcal{J}$ and matrix $\mathbf{F}$ are defined as
\begin{equation}
    \text{Error}(\mathcal{J}) = \dfrac{\lVert \mathcal{J} - \hat{\mathcal{J}} \rVert^2}{\lVert \mathcal{J} \rVert^2}, \; \text{Error}(\mathbf{F}) = \dfrac{\lVert \mathbf{F} - \hat{\mathbf{F}} \rVert^2}{\lVert \mathbf{F} \rVert^2}. \nonumber
\end{equation}
To evaluate how well a decoupled model approximates a function $\mathbf{f}(\mathbf{x})$, a relative root mean-squared error $e_i$ per output $f_i(\mathbf{x})$ is used, defined as a percentage
\begin{equation}
    e_i = \sqrt{\dfrac{\sum^S_{s=1}(f_i(\mathbf{x}^{(s)}) - \widehat{\mathbf{f}}_i(\mathbf{x}^{(s)})^2}{\sum^S_{s=1}(f_i(\mathbf{x}^{(s)}) - \mathbb{E}[f_i])^2}} \times 100, \nonumber
\end{equation}
with $\widehat{f}(\mathbf{x})$ the computed decoupling that approximates $\mathbf{f}(\mathbf{x})$.
%Throughout the experiments, two types of basis functions are used, polynomial and trigonometric. The corresponding sets of basis functions, called Poly and Trig, are defined as
%\begin{gather}
%    \text{Poly} = \{u, u^2, \hdots, u^{d_\ell} \} 
%\end{gather}
%and $\text{Trig} =$
%\[
%\begin{cases} 
% \left\{\cos(u), \sin(u), \cos(2u), \hdots, \cos\left(k u\right) \right\}& d_{\ell} \text{ odd}\\
%\left\{\cos(u), \sin(u), \cos(2u), \hdots, \cos\left(k u\right), \sin\left(k u\right) \right\}& d_{\ell} \text{ even}
%\end{cases}    
%\]
%and $k = \lceil d_\ell/2 \rceil$. For example, 
%if $d_{\ell}$ = 3, the internal function $g^{(j)}_{\ell}$ becomes
%\begin{align}
%    &\text{Poly: } g^{(j)}_{\ell}(u) = c^{(j)}_{\ell, 0} + c^{(j)}_{\ell, 1} u + c^{(j)}_{\ell, 2} u^2 + c^{(j)}_{\ell, 3} u^3, \nonumber \\
%    &\text{Trig: } g^{(j)}_{\ell}(u) = c^{(j)}_{\ell, 0} + c^{(j)}_{\ell, 1} \cos(u) + c^{(j)}_{\ell, 2} \sin(u) + c^{(j)}_{\ell, 3} \cos(2u) \nonumber
%\end{align}
%In each experiment, the same basis functions are used for all internal functions, so either all polynomial or all trigonometric (degrees can vary per layer). 
%The algorithms are implemented in PyTorch.

\subsection{Synthetic examples}
This section shows regression results for two-layer decouplings computed of three synthetic, two-layer examples $\mathbf{f}_1(\mathbf{x})$, $\mathbf{f}_2(\mathbf{x})$ and $\mathbf{f}_3(\mathbf{x})$, shown in \ref{app:synthetic}. The two-layer decouplings are computed using the TS-PROJ-CMPT-$2$ and TS-CONSTR-CMPT-$2$ algorithms, with minimum $10$ and maximum $500$ iterations with execution stopped if the objective function does not decrease for $50$ consecutive iterations. The Jacobian tensor $\mathcal{J}\in \mathbb{R}^{n \times m \times S}$ and zeroth-order information matrix $\mathbf{F} \in \mathbb{R}^{n \times S}$ are constructed using $S=30$ sampling points, sampled uniformly from $[-1,1]^m$. For the outer stopping criteria, a validation set of $30$ sampling points is constructed, also sampled uniformly from $[-1,1]^m$ and the criterion corresponds to the sum of the relative root mean-squared errors on the outputs $e_1 + e_2 + \cdots + e_n$. For each synthetic system, the algorithms compute $30$ two-layer decouplings where each execution samples $S$ new points and a new validation set. The used configurations for computing the two-layer decouplings correspond to the configuration of the actual underlying system, i.e., $\mathbf{f}_1(\mathbf{x})$ has ParaTuck ranks $(2,2)$ with degree-$5$ polynomials in the first layer and degree $2$ polynomials in the second so the algorithms use the same ParaTuck ranks and degrees of polynomials.

\begin{table}[!ht]
    \centering
    \scalebox{0.90}{
    \begin{tabular}{|c|c|c|c|c|c|c|c|}
        \hline
        \multirow{2}{*}{PROJ} & \multicolumn{3}{c|}{Results} & \multirow{2}{*}{CONSTR} & \multicolumn{3}{|c|}{Results} \\ 
        \cline{2-4}\cline{6-8}
        & Mean & Median & Std &  & Mean & Median & Std \\
        \hline
        \hline
        $\mathbf{f}_1(\mathbf{x})$ & \multicolumn{3}{c|}{} & & \multicolumn{3}{c|}{}  \\
        \cline{1-4}\cline{6-8}
        Error($\mathcal{J}$) & $0.0337$ & $0.0138$ & $0.0517$ & & $\mathbf{0.00398}$ & $\mathbf{0.00247}$ & $\mathbf{0.00550}$ \\
        Error($\mathbf{F}$) & $0.0938$ & $0.0565$ & $0.108$ & & $\mathbf{0.0352}$ & $\mathbf{0.0152}$ & $\mathbf{0.0782}$ \\
        \cline{1-4}\cline{6-8}
        $e_1$ (\%) & $\mathbf{0.887}$ & $\mathbf{0.922}$ & $\mathbf{0.178}$ & & $0.997$ & $0.856$ & $0.781$ \\
        $e_2$ (\%) & $\mathbf{0.874}$ & $\mathbf{0.908}$ & $\mathbf{0.173}$ & & $1.02$ & $0.880$ & $0.929$ \\
        \hline
        \hline
        $\mathbf{f}_2(\mathbf{x})$ & \multicolumn{3}{c|}{} & & \multicolumn{3}{c|}{}  \\
        \cline{1-4}\cline{6-8}
        Error($\mathcal{J}$) & $0.00143$ & $0.000412$ & $0.00203$ & & $\mathbf{0.000268}$ & $\mathbf{4.54e-5}$ & $\mathbf{0.000596}$ \\
        Error($\mathbf{F}$) & $0.00862$ & $0.00153$ & $0.0154$ & & $\mathbf{0.00180}$ & $\mathbf{0.000312}$ & $\mathbf{0.00429}$ \\
        \cline{1-4}\cline{6-8}
        $e_1$ (\%) & $\mathbf{0.807}$ & $\mathbf{0.816}$ & $\mathbf{0.239}$ & & $0.904$ & $0.920$ & $0.218$ \\
        $e_2$ (\%) & $\mathbf{0.806}$ & $\mathbf{0.816}$ & $\mathbf{0.240}$ & & $0.904$ & $0.916$ & $0.218$ \\
        $e_3$ (\%) & $\mathbf{0.807}$ & $\mathbf{0.815}$ & $\mathbf{0.239}$ & & $0.904$ & $0.921$ & $0.218$ \\
        \hline
        \hline
        $\mathbf{f}_3(\mathbf{x})$ & \multicolumn{3}{c|}{} & & \multicolumn{3}{c|}{}  \\
        \cline{1-4}\cline{6-8}
        Error($\mathcal{J}$) & $\mathbf{0.00670}$ & $0.00127$ & $\mathbf{0.00979}$ & & $0.00717$ & $\mathbf{0.000498}$ & $0.0179$ \\
        Error($\mathbf{F}$) & $\mathbf{0.0187}$ & $\mathbf{0.00176}$ & $\mathbf{0.0371}$ &  & $0.0509$ & $0.00434$ & $0.100$ \\
        \cline{1-4}\cline{6-8}
        $e_1$ (\%) & $\mathbf{1.03}$ & $\mathbf{1.04}$ & $0.0585$ & & $0.983$ & $1.02$ & $\mathbf{0.483}$ \\
        $e_2$ (\%) & $\mathbf{1.03}$ & $\mathbf{1.04}$ & $0.0572$ & & $0.976$ & $1.03$ & $\mathbf{0.482}$ \\
        $e_3$ (\%) & $\mathbf{1.03}$ & $\mathbf{1.04}$ & $0.0578$ & & $0.979$ & $1.03$ & $\mathbf{0.482}$ \\
        \bottomrule
    \end{tabular}
    }
    \caption{Mean, median and standard deviation (Std) results of the synthetic examples $\mathbf{f}_1(\mathbf{x})$, $\mathbf{f}_2(\mathbf{x})$ and $\mathbf{f}_3(\mathbf{x})$, shown in appendix \ref{app:synthetic}, over $30$ runs of the TS-PROJ-CMPT-$2$ and TS-CONSTR-CMPT-$2$ algorithms. The best mean, median and Std results for each row are highlighted in bold.}
    \label{tab:synth_example}
\end{table}

Table \ref{tab:synth_example} shows the mean, median and standard deviation (Std) of the Error($\mathcal{J}$), Error($\mathbf{F}$) and validation set errors on the outputs, for the TS-PROJ-... and TS-CONSTR-... algorithms on the three synthetic examples. For $\mathbf{f}_1(\mathbf{x})$ and $\mathbf{f}_2(\mathbf{x})$, the Error($\mathcal{J}$) and Error($\mathbf{F}$) results of the CONSTR algorithm show an improvement of about an order of magnitude compared to the PROJ algorithm, while for $\mathbf{f}_3(\mathbf{x})$ the results are comparable. The results for the errors on the outputs are comparable for the three examples for both algorithms, yielding mean and median errors around $1\%$. The only noticeable difference in the output errors between the algorithms is in the standard deviation of the PROJ algorithm for $\mathbf{f}_1(\mathbf{x})$ being $0.2$ compared to $\approx 0.8$ for the CONSTR algorithm. 
Overall, the results in Table \ref{tab:synth_example} show that for the three synthetic examples both algorithms are performing well.

\subsection{System identification: Silverbox benchmark}
Similar to \cite{decuyper2022decoupling}, we will fit a Polynomial NARX (PNARX) model to the data of the Silverbox benchmark~\cite{wigren2013three}. 
In general, a NARX model aims to estimate the output $y(t)$ at timepoint $t$ based on previous output estimates $y(t-i)$, for $i=1,2,\hdots,n_y$ as well as the current and previous inputs $x(t-j)$, for $j=0,1,2,\hdots,n_x$. 
Algebraically this can be stated as
\begin{equation}
    y(t) = F(x(t), x(t-1), \hdots, x(t-n_x), y(t-1), \hdots, y(t-n_y)) + \epsilon_t, \nonumber
\end{equation}
where, in the case of PNARX models, the model representation F(.) corresponds to a multivariate polynomial.

For fitting the Silverbox benchmark data, a polynomial basis of degree $3$ is used with $n_x=1$, $n_y=3$ and the inclusion of all monomials, leading to a model with $55$ parameters. This setup corresponds to the PNARX configuration used in \cite{decuyper2022decoupling} for the Silverbox benchmark to allow for direct comparison of decoupling results. The reference PNARX model trained for this work, with which the decouplings are compared, has a prediction error of $0.88\%$ on the test set. Finally, the Silverbox benchmark contains a set of training data, of which $S=200$ points (uniformly sampled without replacement) are used to construct $\mathcal{J} \in \mathbb{R}^{1 \times 5 \times 200}$ and the matrix $\mathbf{F}^{1 \times 200}$. For evaluation, the decouplings are assessed on test data from the Silverbox benchmark in two ways: model approximation error and decoupling prediction error. 
\\\\
\textbf{Model approximation error}: How well does the output of the decoupling match the output of the reference PNARX model in each point of the test set. This corresponds to evaluating the decoupling in a regression context of approximating the function represented by the PNARX model.
\\\\
\textbf{Decoupling prediction error}: How does the decoupling perform when comparing the predicted outputs of the decoupling, to the Silverbox benchmark test data. This corresponds to evaluating the decoupling in a system identification model prediction context where there is direct comparison with the test data, not the PNARX model output (as is the case for the model approximation).
\\\\
Both the model approximation and decoupling prediction error are represented by the relative root-mean squared error (RRMSE).

Figure \ref{fig:silverbox_results} shows the results for the Silverbox experiment with 30 executions of the respective algorithms per data point. In addition, for each run a new Jacobian tensor $\mathcal{J}$ and zeroth-order information matrix $\mathbf{F}$ are constructed using $S=200$ sample points and for the \textit{Zniyed} algorithm $\lambda$ is set to $1e-4$. For the outer stopping criterion of the \textit{TS-} algorithms, the model approximation error on the validation set is used which is the full training set with the samples used for construction of $\mathcal{J}$ and $\mathbf{F}$ omitted.

The results per algorithm are divided columnwise where the model approximation and simulation error for \textit{Zniyed}, \textit{TS-Zniyed} and \textit{TS-PROJ-CMPT-$1$} are shown in the left column (considering that the TS-PROJ-CMPT-$L$ algorithm can be seen as a generalization of the \textit{Zniyed} CMTF algorithm \cite{zniyed2021tensor}) and the results for the \textit{Hollander} and \textit{TS-CONSTR-CMPT-$1$} algorithms are shown in the right column (since the CONSTR-CMPT-$L$ algorithm can be seen as a generalization of the polynomial constraint algorithm of Hollander \cite{hollander2017multivariate}). 
\begin{figure}[ht!]
    \centering
    \subfloat{
        \includegraphics[width=0.45\textwidth]{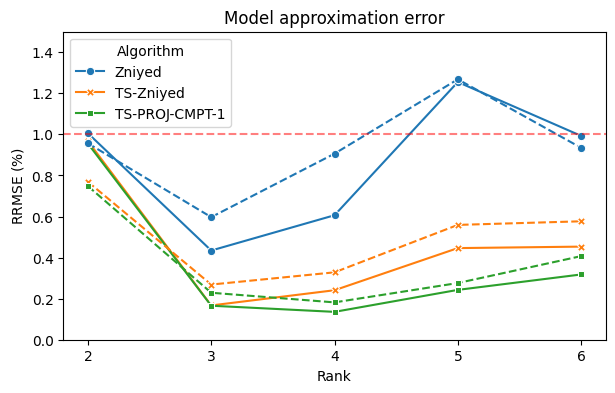}}
    \hfill
    \subfloat{
        \includegraphics[width=0.45\textwidth]{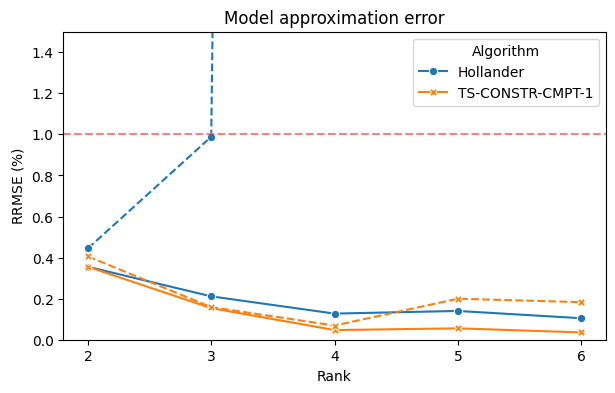}}
        \hfill
    \subfloat{
        \includegraphics[width=0.45\textwidth]{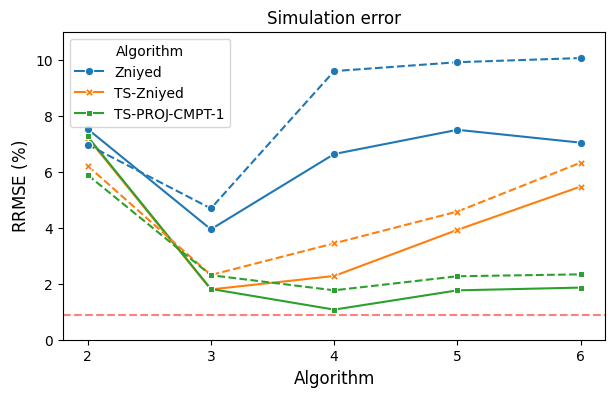}}
    \hfill
    \subfloat{
        \includegraphics[width=0.45\textwidth]{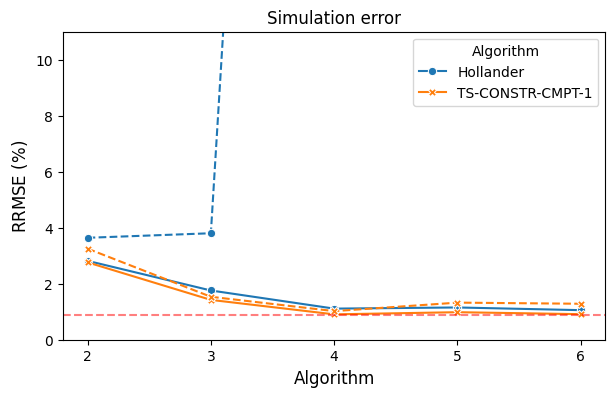}}
    \caption{Median (solid) and mean (dotted) values over $30$ runs for the model approximation (first row) and simulation error (second row) for the decouplings computed of the reference PNARX model on the Silverbox benchmark using the \textit{Zniyed}, \textit{TS-Zniyed} and \textit{TS-PROJ-CMPT-1} algorithms (first column) and the \textit{Hollander} and \textit{TS-CONSTR-CMPT-1} algorithms (second column). Decouplings are computed for ranks $2$ up until $6$. For the model approximation error, the red dotted line denotes an error of $1\%$ and for the simulation error, the red dotted line denotes the performance of the original model of $0.88\%$. Overall, the \textit{TS-PROJ-...} and \textit{TS-CONSTR-...} methods yield the best model approximation and simulation error for their respective comparisons, as well as the most stable results.}
    \label{fig:silverbox_results}
\end{figure}

The left column of Figure~\ref{fig:silverbox_results} shows that the two stage algorithms \textit{TS-Zniyed} and \textit{TS-PROJ-CMPT-$1$} give better and more stable results compared to the standard \textit{Zniyed} results, i.e., lower median values and smaller distance between the mean and median data points. In addition, the \textit{TS-PROJ-CMPT-$1$} algorithm seems to be overall best performing with the interesting observation that the RRMSE remains quite stable for ranks $3$ through $6$, while for the other algorithms the RRMSE starts to noticeably increase as the rank increases from $2$ through $6$. For the right column of Figure~\ref{fig:silverbox_results} the two stage algorithms only slightly improve results compared to the standard \textit{Hollander} algorithm. However, due to the incorporation of the zeroth-order information matrix $\mathbf{F}$, the \textit{TS-CONSTR-CMPT-$1$} algorithm is much more stable with the mean and median values almost coinciding while for the \textit{Hollander} algorithm the gap between the mean and median is much larger due to possible faulty bias corrections. Thus, the \textit{TS-PROJ-CMPT-$1$} and \textit{TS-CONSTR-CMPT-$1$} algorithms introduced in this work provide the best results in this comparison.

Comparing the results of both columns it is quite clear that for this example, the algorithms in the right column are the best fit for this problem (lower and consistent median values) with \textit{TS-CONSTR-CMPT-1} providing the best results. Furthermore, the results in Figure~\ref{fig:silverbox_results} are comparable with those achieved by the non-parametric decoupling algorithm of Decuyper et al.~\cite{decuyper2022decoupling}. 
Do note however, that 1) the algorithms introduced in this work require only a value for the ranks, degree and $\beta$ and $\lambda$ parameter to be set, while for the algorithm of Decuyper~\cite{decuyper2022decoupling} additional terms and filters may need to be added/removed from the objective function to achieve satisfactory results, which is more cumbersome and 2) our results are statistics over $30$ runs per data point, while the behavior and stability over multiple runs are unclear in the results of Decuyper~\cite{decuyper2022decoupling}.

\subsection{Neural network case one: MNIST}

 \noindent In this section, a fully connected neural network, trained on the MNIST dataset~\cite{deng2012mnist}
 is used. The network has an input layer of size $784$, three hidden layers of size $80$, $60$ and $40$ with ReLU activation functions and an output layer of size $10$ with a Sigmoid activation. The network has an accuracy of $95.61 \%$ on the test set.
 
For computational reasons, the \textit{Hollander} and \textit{CONSTR} type algorithms will not be used for the neural network experiments and the focus will be on \textit{Zniyed} and \textit{PROJ} type algorithms.

The decoupling method can be used to replace a part of this pre-trained network with a decoupled representation, e.g., for reducing model size. 
Figure~\ref{fig:MNIST_net_compr} shows the architecture and part of the network that will be compressed. The subnet to be compressed corresponds to a multivariate function with $80$ inputs and $10$ outputs. Additionally, the \textit{savings percentage} (SP) and \textit{accuracy drop} (AD) metrics, defined as
\begin{gather}
    \text{SP} = \left(1 - \dfrac{\# \text{Parameters in decoupled representation}}{\# \text{Parameters in original system}}\right) \times 100, \nonumber \\
    \text{AD} = acc_{B} - acc_{A}, \nonumber
\end{gather}
are used for analyzing the model size reduction. Here, $acc_{B}$ and $acc_{A}$ denote the accuracy of the network on the test set before and after the model size reduction respectively.  The SP denotes the percentage of parameters that we save (or leave out) by using the decoupled representation, i.e., higher SP equals more compression.

For each execution of the algorithms, a new Jacobian tensor $\mathcal{J}$ with $S=200$ sampling points is used. The sampling points are distributed equally over the different output classes, so given that there are $10$ output classes, this means $20$ samples per MNIST class. In addition, the internal functions of the decoupling are polynomials of degree $4$. For the $stop\_metric$, a validation set is constructed by removing the $200$ sample points of the Jacobian tensor from the training set. The $stop\_metric$ then corresponds to the accuracy on this validation set. As a result, the stopping criterion in Algorithm~\ref{alg:practical_use_algorithm} can be seen as a form of early stopping based on the validation set. Finally, each algorithm is executed for a minimum of $15$, maximum $50$ iterations with early stoppage if there is no decrease in the objective for $20$ consecutive iterations and the \textit{Zniyed} algorithm is executed with the same configuration as described by Zniyed \cite{zniyed2021tensor}, i.e., $\lambda$ is initially set to $\lambda = 1e-6$ and multiplied by $10$ every $10$ iterations.
\begin{figure}[ht]
    \centering
    \includegraphics[width=0.33\textwidth]{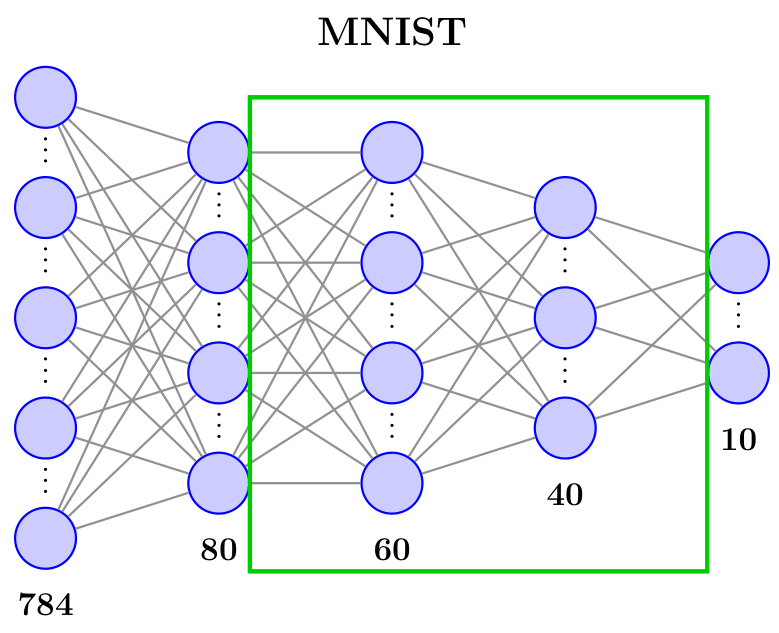}
    \caption{Architecture of MNIST network. The section in the green box is compressed.}
    \label{fig:MNIST_net_compr}
\end{figure}

\begin{figure}
    \centering
    \includegraphics[width= 0.95\textwidth]{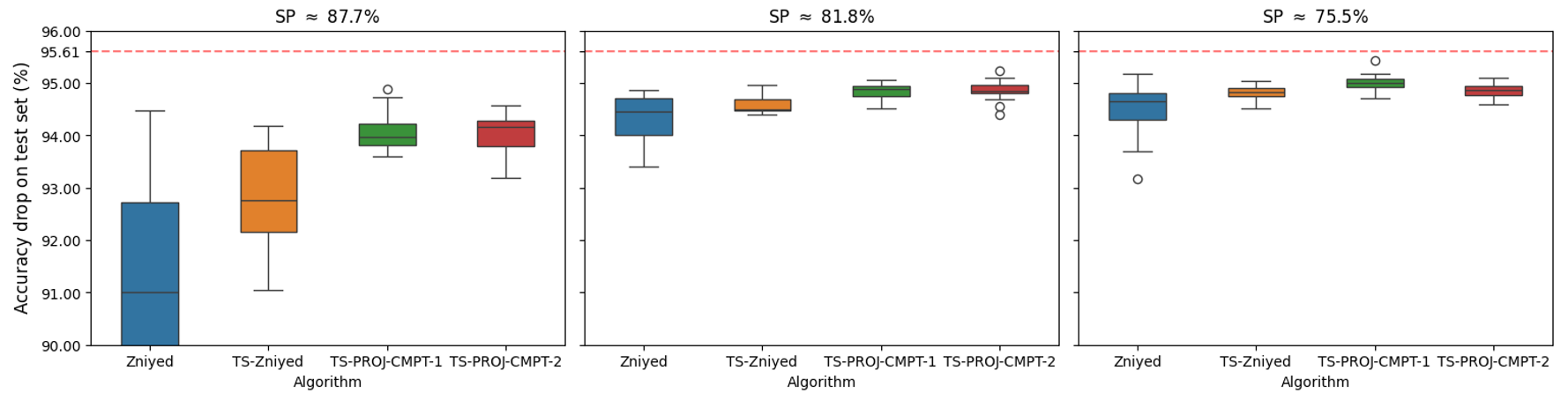}
    \caption{Accuracies on the MNIST test set after compressing the original network by replacing the final two hidden layers with a decoupled representation using algorithms mentioned on the x-axis. For \textit{Zniyed}, \textit{TS-Zniyed} and \textit{TS-PROJ-CMPT-$1$}, $r=10,15,20$ from left to right and for \textit{TS-PROJ-CMPT-$2$} we have $(r_2, r_1) = (12, 8), (11, 13)$ and $(11,18)$ from left to right. The title above each plot shows the corresponding savings percentage of the computed decouplings and the internal functions are polynomials of degree four. Each boxplot consists of $15$ results. The red dotted lines denote the accuracy on the test set of the original network. The stability and performance of the \textit{Zniyed} executions improves as the SP decreases, but the \textit{TS-PROJ-...} executions provide the smallest and most stable accuracy drops.}
    \label{fig:zniyed_comparison}
\end{figure}
%Figure \ref{fig:decoupling_vs_standard} indicates that for this example, the Adam algorithm requires seven times more data ($1400$ training points) than the tensor-based method ($S=200$) to achieve better results.

Figure~\ref{fig:zniyed_comparison} shows a comparison of $15$ single-layer decouplings computed with the algorithms \textit{Zniyed}, \textit{TS-Zniyed}, \textit{TS-PROJ-CMTPT-$1$} (single-layer) and \textit{TS-PROJ-CMPT-$2$} (two-layer), for three different CR values that are shown in the title of the plots. Figure \ref{fig:zniyed_comparison} shows that compared with the results of the Zniyed and TS-Zniyed algorithms, the newly introduced TS-PROJ-... algorithms provide more stable and better results, both with single and two-layer decouplings. Furthermore, the results show that the accuracy results increase for the single and two-layer decouplings when the savings percentage drops from $\text{SP} \approx 87.7\%$ to $\text{SP} \approx 81.8\%$ and $\text{SP} \approx 75.5\%$, yielding AD's lower than $0.6\%$ for $\text{SP} \approx 81.8\%$ and $\text{SP}  \approx 75.5\%$.

\subsection{Neural network case two: FashionMNIST}
\noindent Since MNIST in general is not a hard problem, this section looks at a fully connected neural network trained on a more difficult dataset, namely FashionMNIST \cite{xiao2017fashion}. The network has an input layer of size $784$, four hidden layers of sizes $512$, $256$, $128$, and $64$ with ReLU activation functions, and an output layer of size $10$ with a logarithmic softmax activation. The network has an accuracy of $89.14 \%$ on the test set. Again, part of the pre-trained network will be replaced with a decoupled representation to reduce the model size. Figure \ref{fig:FashionMNIST_net_compr} shows the architecture of the network and the part that will be replaced. Here, the subnet to be compressed is a multivariate function with $512$ inputs and $10$ outputs.

Similarly to the MNIST example, a new Jacobian tensor is constructed for each execution of the respective algorithms and the internal functions are polynomials of degree $4$. The configuration of algorithm executions as well as the construction of the validation set and $stop\_metric$ are analogous to the MNIST example.

\begin{figure}[ht]
    \centering
    \includegraphics[width=0.4\textwidth]{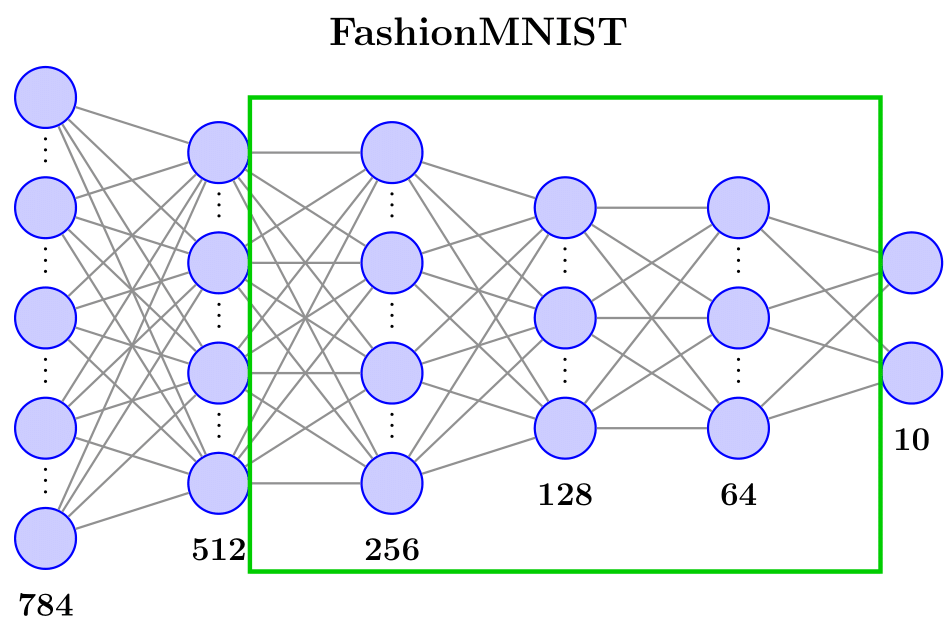}
    \caption{Architecture of FashionMNIST network. The section in the green box is compressed.}
    \label{fig:FashionMNIST_net_compr}
\end{figure}

Figure \ref{fig:FashionMNIST_Acc_full} shows compression results for algorithms \textit{Zniyed}, \textit{TS-Zniyed}, \textit{TS-PROJ-CMPT-$1$} and \textit{TS-PROJ-CMPT-$2$} with decreasing savings percentage from left to right. This figure shows the results for $5$ executions of the algorithms with the different configurations and $S=200$ sample points.

\begin{figure}[ht]
    \centering
    \includegraphics[width=0.75\textwidth]{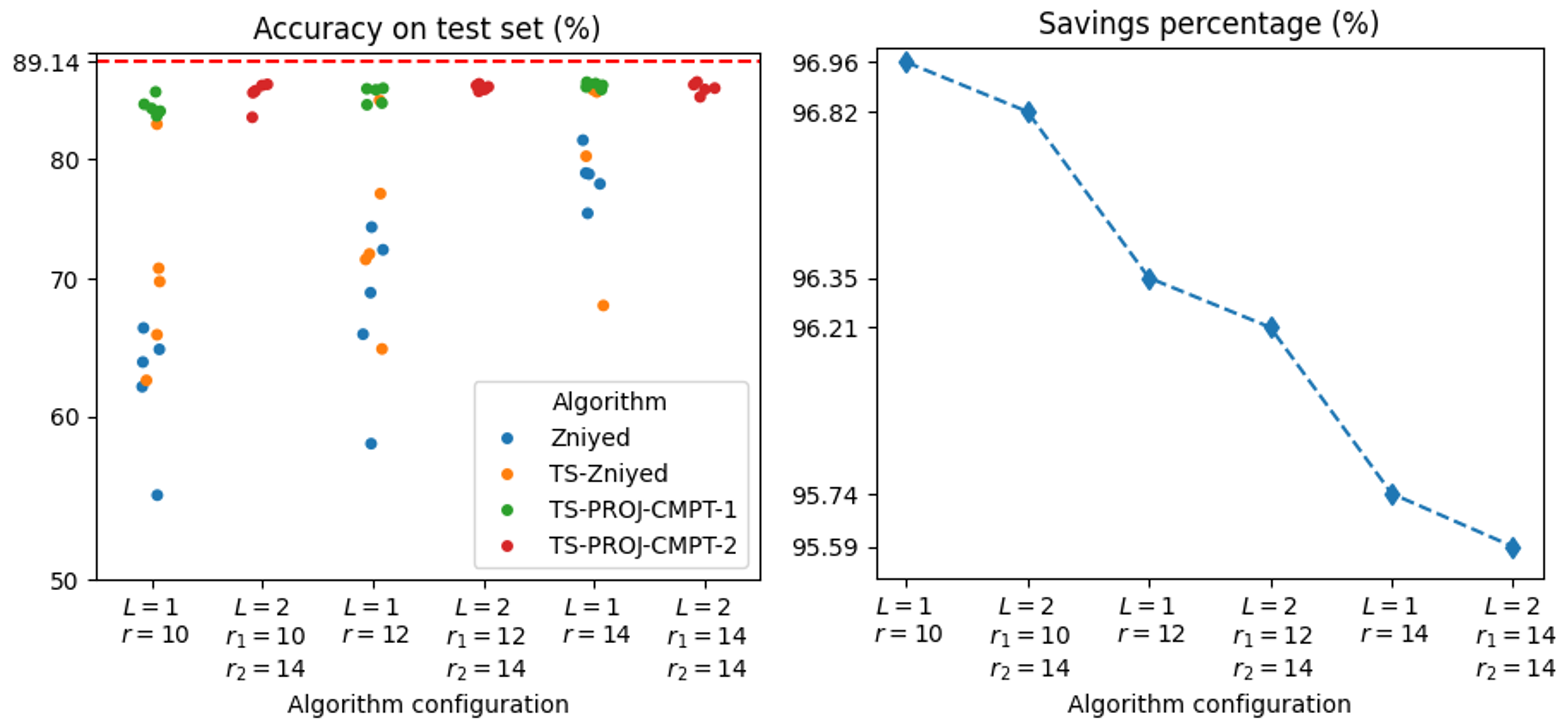}
    \caption{ (Left) Accuracy results on the test set for the networks compressed using the
configurations mentioned on the x-axis and algorithms \textit{Zniyed}, \textit{TS-Zniyed}, \textit{TS-PROJ-CMPT-$1$} and \textit{TS-PROJ-CMPT-$2$}. The red dotted line indicates the accuracy of the
original network on the test set. (Right) The corresponding SP values of the different
configurations used for the decoupling. For the used algorithm configurations, the Zniyed executions yield very large and unstable accuracy drops. Compared to \textit{Zniyed}, the results of the \textit{TS-PROJ-...} executions show smaller and more stable accuracy drops with a less noticeable increase in accuracy drop as the SP increases.
}
    \label{fig:FashionMNIST_Acc_full}
\end{figure}

The results in Figure~\ref{fig:FashionMNIST_Acc_full} show strong and stable compression results for the \textit{TS-PROJ-...} algorithms compared to the \textit{Zniyed} variants which are unstable and erratic. Compression of the network is possible using the decoupling algorithms where only $\approx 3.8\%$ of the original amount of parameters in the sub-network are kept, while having an AD of less than $2.5\%$ on the test set. 
Given the SP, this is quite substantial using only a limited amount of data.

\begin{figure}[ht]
    \centering
    \includegraphics[width=0.42\textwidth]{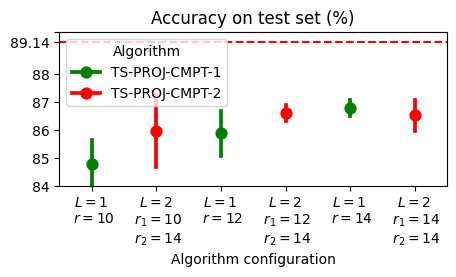}
    \caption{Accuracy results on the test set for the networks compressed using the
configurations mentioned on the x-axis and \textit{TS-PROJ-CMPT-$1$} and \textit{TS-PROJ-CMPT-$2$} algorithms, together with bars indicating $\pm$ one standard deviation. The red dotted line indicates the accuracy of the original network on the test set. For the higher SP values (left side) it can be seen that the two-layer decoupling, computed by the \textit{TS-PROJ-CMPT-2} algorithm yields a better accuracy drop compared to the $L=1, r=10$, single-layer representation (that has a similar SP value) and achieves comparable accuracy drops to the $L=1, r=12$, single-layer configuration while having a $\approx 0.5\%$ higher SP value.}
    \label{fig:FashionMNIST_Condensed}
\end{figure}

Furthermore, Figure~\ref{fig:FashionMNIST_Condensed} focuses on the results of the \textit{TS-PROJ-...} algorithms, showing that for this experiment, the used two-layer decoupling configurations are able to reach better and more stable accuracies on the test set, compared to some single-layer configurations, while only marginally decreasing the SP. Compare for example the two leftmost configurations in the figure. This indicates that a two-layer decoupling can be advantageous when considering very high SP values and when adding a second layer in the decoupled model yields a slight decrease in SP, but improves the flexibility of the model.

%\section{Conclusions}
%\textcolor{red!80!black}{Todo}
\section{Conclusions}

\noindent This work extends the tensor-based framework for computing single-layer decouplings by 1) introducing a generalized theory for the concept of multi-layer decoupling, 2) providing a theoretical basis for choosing analytic and scale-invariant basis functions, 3) deriving a tensor-based solution strategy based on \cite{dreesen2015decoupling}, introducing the use of the ParaTuck-$L$ decomposition, 4) describing the PROJ-CMPT-$L$ and CONSTR-CMPT-$L$ algorithms for computing multi-layer decouplings, and 5) introducing a two-stage stopping criterion for controlling the non-monotonicity of the algorithms and taking into account the problem context. 

As illustrated by both synthetic and FashionMNIST examples, the two-layer decoupling may result in a better approximation than its single-layer alternative.
Furthermore, for the Silverbox example, the CONSTR-CMPT-$L$ algorithm proved to be a better choice than the PROJ-CMPT-$L$ algorithm while for the neural network case studies (with larger tensors to decompose), the PROJ-CMPT-$L$ algorithm is preferred for computational reasons.

For the neural network case studies, applying PROJ-CMPT-$L$ for compressing an MNIST (simple) and FashionMNIST (harder) network, led to the following two conclusions. The PROJ-CMPT-$L$ algorithm, executed with the bilevel optimization strategy, yields more stable results than the CMTF algorithm proposed by \cite{zniyed2021tensor}, even when it is analogously executed with bilevel optimization. Lastly, good compression results are achieved on the MNIST and FashionMNIST networks, with the FashionMNIST case in particular yielding SP values of $\approx96.2\%$ with less than a $2.5\%$ drop in accuracy on the test set. 

In summary, our results provide a theoretically supported extension of the single-layer decoupling problem to the multi-layer case with initial algorithms applied to different problem contexts (system identification, neural networks), showcasing wide applicability of the developed methods and theory.
Future work includes formulating a strategy for choosing the number of neurons in multi-layer decouplings, studying additional basis functions for the internal functions, improving on the developed algorithms to better incorporate large amounts of data or structures with more than two layers together with analyzing 
their potential advantages over the single- and two-layer case.

\section{Acknowledgements}
\noindent
This work was supported by the Fonds Wetenschappelijk Onderzoek (FWO) fundamental research fellowship 11A2H25N and partly supported by the funding from French government, managed by the National Research Agency (ANR), under the France 2030 program, reference ANR-25-PEIA-0003.

\bibliographystyle{elsarticle-num}
\bibliography{refs}

@article{dreesen2015decoupling,
  title={Decoupling multivariate polynomials using first-order information and tensor decompositions},
  author={Dreesen, Philippe and Ishteva, Mariya and Schoukens, Johan},
  journal_={SIAM Journal on Matrix Analysis and Applications},
  journal={SIAM J Matrix Analysis and Applications},
  volume={36},
  number_={2},
  pages={864--879},
  year={2015},
  publisher={SIAM}
}

@inproceedings{dreesen2015decoupling_wiener,
  title={Decoupling static nonlinearities in a parallel {W}iener-{H}ammerstein system: A first-order approach},
  author={Dreesen, Philippe and Schoukens, Maarten and Tiels, Koen and Schoukens, Johan},
  booktitle={IEEE International Instrumentation and Measurement Technology Conference (I2MTC) Proceedings},
  pages={987--992},
  year={2015},
  organization={IEEE}
}

@inproceedings{dreesen2016decoupling,
  title={Decoupling nonlinear state-space models: case studies},
  author={Dreesen, Philippe and Esfahani, A Fakhrizadeh and Stoev, Julian and Tiels, Koen and Schoukens, Johan},
  booktitle={Proceedings of the International Conference on Noise and Vibration Engineering (ISMA)},
  pages={2639--2646},
  year={2016}
}

@inproceedings{noel2016hysteretic,
  title={Hysteretic benchmark with a dynamic nonlinearity},
  author={Noel, Jean-Philippe and Schoukens, Maarten},
  booktitle={Workshop on nonlinear system identification benchmarks},
  pages={7--14},
  year={2016}
}

@inproceedings{wigren2013three,
  title={Three free data sets for development and benchmarking in nonlinear system identification},
  author={Wigren, Torbj{\"o}rn and Schoukens, Johan},
  booktitle={European control conference (ECC)},
  pages={2933--2938},
  year={2013},
  organization={IEEE}
}

@article{decuyper2019decoupling,
  title={Decoupling multivariate polynomials for nonlinear state-space models},
  author={Decuyper, Jan and Dreesen, Philippe and Schoukens, Johan and Runacres, Mark C and Tiels, Koen},
  journal={IEEE Control Systems Letters},
  volume={3},
  number={3},
  pages={745--750},
  year={2019},
  publisher={IEEE}
}

@article{decuyper2021decoupling,
  title={Decoupling {P-NARX} models using filtered {CPD}},
  author={Decuyper, Jan and Westwick, David and Karami, Kiana and Schoukens, Johan},
  journal={IFAC-PapersOnLine},
  volume={54},
  number={7},
  pages={661--666},
  year={2021},
  publisher={Elsevier}
}

@article{dreesen2021parameter,
  title={Parameter estimation of parallel Wiener-Hammerstein systems by decoupling their Volterra representations},
  author={Dreesen, Philippe and Ishteva, Mariya},
  journal={IFAC-PapersOnLine},
  volume={54},
  number={7},
  pages={457--462},
  year={2021},
  publisher={Elsevier}
}

@article{decuyper2022decoupling,
  title={Decoupling multivariate functions using a nonparametric filtered tensor decomposition},
  author={Decuyper, Jan and Tiels, Koen and Weiland, Siep and Runacres, Mark C and Schoukens, Johan},
  journal={Mechanical Systems and Signal Processing},
  volume={179},
  pages={109328},
  year={2022},
  publisher={Elsevier}
}

@inproceedings{hollander2016parallel,
  title={Parallel Wiener-Hammerstein Identification: A case study},
  author={Hollander, Gabriel and Dreesen, Philippe and Ishteva, Mariya and Schoukens, Johan},
  booktitle={ISMA2016 International Conference on Noise and Vibration Engineering and USD2016 International Conference on Uncertainty on Structural Dynamics},
  pages={2647--2656},
  year={2016},
  organization={KU~Leuven}
}

@article{karami2021applying,
  title={Applying polynomial decoupling methods to the polynomial {NARX} model},
  author={Karami, Kiana and Westwick, David and Schoukens, Johan},
  journal={Mechanical Systems and Signal Processing},
  volume={148},
  pages={107134},
  year={2021},
  publisher={Elsevier}
}

@inproceedings{zniyed2021tensor,
  title={A tensor-based approach for training flexible neural networks},
  author={Zniyed, Yassine and Usevich, Konstantin and Miron, Sebastian and Brie, David},
  booktitle={55th Asilomar Conference on Signals, Systems, and Computers},
  pages={1673--1677},
  year={2021},
  organization={IEEE}
}

@article{apicella2021survey,
  title={A survey on modern trainable activation functions},
  author={Apicella, Andrea and Donnarumma, Francesco and Isgr{\`o}, Francesco and Prevete, Roberto},
  journal={Neural Networks},
  volume={138},
  pages={14--32},
  year={2021},
  publisher={Elsevier}
}

@inproceedings{de2023compressing,
  title={Compressing Neural Networks with Two-Layer Decoupling},
  author={De Jonghe, Joppe and Usevich, Konstantin and Dreesen, Philippe and Ishteva, Mariya},
  booktitle={IEEE 9th International Workshop on Computational Advances in Multi-Sensor Adaptive Processing (CAMSAP)},
  pages={226--230},
  year={2023},
  organization={IEEE}
}

@phdthesis{hollander2017multivariate,
  title={Multivariate polynomial decoupling in nonlinear system identification},
  author={Hollander, Gabriel},
  year={2017},
  school={Vrije Universiteit Brussel (VUB)},
}

@book{liu2021tensors,
  title={Tensors for Data Processing: Theory, Methods, and Applications},
  author={Liu, Yipeng},
  year={2021},
  publisher={Academic Press}
}

@article{kolda2009tensor,
  title={Tensor decompositions and applications},
  author={Kolda, Tamara G and Bader, Brett W},
  journal={SIAM Review},
  volume={51},
  number={3},
  pages={455--500},
  year={2009},
  publisher={SIAM}
}

@article{harshman1996uniqueness,
  title={Uniqueness proof for a family of models sharing features of {T}ucker's three-mode factor analysis and PARAFAC/CANDECOMP},
  author={Harshman, Richard A and Lundy, Margaret E},
  journal={Psychometrika},
  volume={61},
  pages={133--154},
  year={1996},
  publisher={Springer}
}

@inproceedings{usevich2023tensor,
  title={Tensor-based two-layer decoupling of multivariate polynomial maps},
  author={Usevich, Konstantin and Zniyed, Yassine and Ishteva, Mariya and Dreesen, Philippe and de Almeida, Andr{\'e} LF},
  booktitle={2023 31st European Signal Processing Conference (EUSIPCO)},
  pages={655--659},
  year={2023},
  organization={IEEE}
}

@inproceedings{dreesen2018decoupling,
  title={Decoupling multivariate functions using second-order information and tensors},
  author={Dreesen, Philippe and De Geeter, Jeroen and Ishteva, Mariya},
  booktitle_={Latent Variable Analysis and Signal Separation: 14th International Conference, LVA/ICA, Guildford, UK, July 2--5, 2018, Proceedings 14},
  booktitle={Latent Variable Analysis and Signal Separation: 14th International Conference, LVA/ICA, Guildford, UK, July 2--5, 2018},
  pages={79--88},
  year={2018},
  organization={Springer}
}

@article{zniyed2021learning,
  title={Learning nonlinearities in the decoupling problem with structured {CPD}},
  author={Zniyed, Yassine and Usevich, Konstantin and Miron, Sebastian and Brie, David},
  journal={IFAC-PapersOnLine},
  volume={54},
  number={7},
  pages={685--690},
  year={2021},
  publisher={Elsevier}
}

@book{giri2010block,
  title = {Block-Oriented Nonlinear System Identification},
  author={Giri, Fouad and Bai, Er-Wei},
  series = {Lecture Notes in Control and Information Sciences},
  volume={404},
  year={2010},
  publisher={Springer}
}

@article{deng2012mnist,
  title={The {MNIST} database of handwritten digit images for machine learning research},
  author={Deng, Li},
  journal={IEEE Signal Processing Magazine},
  volume={29},
  number={6},
  pages={141--142},
  year={2012},
  publisher={IEEE}
}

@Article{shi2016dpn,
  author   = {Jun Shi and Shichong Zhou and Xiao Liu and Qi Zhang and Minhua Lu and Tianfu Wang},
  title    = {Stacked deep polynomial network based representation learning for tumor classification with small ultrasound image dataset},
  journal  = {Neurocomputing},
  year     = {2016},
  volume   = {194},
  pages    = {87-94},
  issn     = {0925-2312},
  }

@article{schoukens2019nonlinear,
  title={Nonlinear system identification: A user-oriented road map},
  author={Schoukens, Johan and Ljung, Lennart},
  journal={IEEE Control Systems Magazine},
  volume={39},
  number={6},
  pages={28--99},
  year={2019},
  publisher={IEEE}
}

@article{janzamin2015beating,
  title={Beating the perils of non-convexity: Guaranteed training of neural networks using tensor methods},
  author={Janzamin, Majid and Sedghi, Hanie and Anandkumar, Anima},
  journal={arXiv preprint arXiv:1506.08473},
  year={2015}
}

@article{comon2006blind,
  title={Blind identification of under-determined mixtures based on the characteristic function},
  author={Comon, Pierre and Rajih, Myriam},
  journal={Signal Processing},
  volume={86},
  number={9},
  pages={2271--2281},
  year={2006},
  publisher={Elsevier}
}

@article{fornasier2021robust,
  title={Robust and resource efficient identification of shallow neural networks by fewest samples.},
  author={Fornasier, Massimo and Vyb{\'\i}ral, Jan and Daubechies, Ingrid},
  journal={Information \& Inference: A Journal of the IMA},
  volume={10},
  number={2},
  year={2021}
}

@Article{comon2017identifiability,
  author  = {Comon, Pierre and Qi, Yang and Usevich, Konstantin},
  title   = {Identifiability of an {X}-Rank Decomposition of Polynomial Maps},
  journal = {SIAM Journal on Applied Algebra and Geometry},
  year    = {2017},
  volume  = {1},
  number  = {1},
  pages   = {388-414}
}

@article{xiao2017fashion,
  title={Fashion-{MNIST}: a novel image dataset for benchmarking machine learning algorithms},
  author={Xiao, Han and Rasul, Kashif and Vollgraf, Roland},
  journal={arXiv preprint arXiv:1708.07747},
  year={2017}
}

@article{boulle2020rational,
  title={Rational neural networks},
  author={Boull{\'e}, Nicolas and Nakatsukasa, Yuji and Townsend, Alex},
  journal={Advances in Neural Information Processing Systems},
  volume={33},
  pages={14243--14253},
  year={2020}
}

@inproceedings{telgarsky2017neural,
  title={Neural networks and rational functions},
  author={Telgarsky, Matus},
  booktitle={International Conference on Machine Learning},
  pages={3387--3393},
  year={2017},
  organization={PMLR}
}

@inproceedings{ yang2024kolmogorov,
  title={{Kolmogorov}-{Arnold} Transformer},
  author={Yang, Xingyi and Wang, Xinchao},
  booktitle={The Thirteenth International Conference on Learning Representations (ICLR)},
  year={2025},
  note={arXiv preprint arXiv:2409.10594},
}

@article{unser2019representer,
  title={A representer theorem for deep neural networks},
  author={Unser, Michael},
  journal={Journal of Machine Learning Research},
  volume={20},
  number={110},
  pages={1--30},
  year={2019}
}

@inproceedings{balestriero2018spline,
  title={A spline theory of deep learning},
  author={Balestriero, Randall and Baraniuk, Richard G.},
  booktitle={International Conference on Machine Learning},
  pages={374--383},
  year={2018},
  organization={PMLR}
}

@article{parhi2021banach,
  title={Banach space representer theorems for neural networks and ridge splines},
  author={Parhi, Rahul and Nowak, Robert D},
  journal={Journal of Machine Learning Research},
  volume={22},
  number={43},
  pages={1--40},
  year={2021}
}

@article{bohra2020learning,
  title={Learning activation functions in deep (spline) neural networks},
  author={Bohra, Pakshal and Campos, Joaquim and Gupta, Harshit and Aziznejad, Shayan and Unser, Michael},
  journal={IEEE Open Journal of Signal Processing},
  volume={1},
  pages={295--309},
  year={2020},
  publisher={IEEE}
}

@inproceedings{molinapade,
  title={Pad{\'e} Activation Units: End-to-end Learning of Flexible Activation Functions in Deep Networks},
  author={Molina, Alejandro and Schramowski, Patrick and Kersting, Kristian},
  booktitle={International Conference on Learning Representations},
  year = {2000}
}

@article{sitzmann2020implicit,
  title={Implicit neural representations with periodic activation functions},
  author={Sitzmann, Vincent and Martel, Julien and Bergman, Alexander and Lindell, David and Wetzstein, Gordon},
  journal={Advances in Neural Information Processing Systems},
  volume={33},
  pages={7462--7473},
  year={2020}
}

@book{lorentz2012bernstein,
  title={Bernstein polynomials},
  author={Lorentz, George G},
  year={2012},
  publisher={American Mathematical Soc.}
}

@book{iarrobino1999power,
  title={Power sums, Gorenstein algebras, and determinantal loci},
  author={Iarrobino, Anthony and Kanev, Vassil},
  year={1999},
  publisher={Springer Science \& Business Media}
}

@book{landsberg2011tensors,
  title={Tensors: geometry and applications: geometry and applications},
  author={Landsberg, Joseph M},
  volume={128},
  year={2011},
  publisher={American Mathematical Society}
}

@inproceedings{dreesen2018lvaica,
  author = {P. Dreesen and J. {De Geeter} and M. Ishteva},
  title = {Decoupling multivariate functions using second-order information and tensors},
  booktitle = {Proc. 14th International Conference on Latent Variable Analysis and Signal Separation (LVA/ICA 2018)},
  series = {Lecture Notes on Computer Science (LNCS)},
  volume = {10891},
  editor = {Y. Deville and S. Gannot and R. Mason and M.~D. Plumbley and D. Ward},
  pages = {79--88},
  address = {Guildford, UK},
  year = {2018},
  publisher = {Springer}
}

@article{ delathauwer2006linkCPDsimul,
author = {{De Lathauwer}, Lieven},
title = {A Link between the Canonical Decomposition in Multilinear Algebra and Simultaneous Matrix Diagonalization},
journal = {SIAM Journal on Matrix Analysis and Applications},
volume = {28},
number = {3},
pages = {642-666},
year = {2006},
}

@article{de2019paratuck,
  title={{PARATUCK} semi-blind receivers for relaying multi-hop MIMO systems},
  author={de Oliveira, Pedro Marinho R and Fernandes, C Alexandre Rolim and Favier, G{\'e}rard and Boyer, R{\'e}my},
  journal={Digital Signal Processing},
  volume={92},
  pages={127--138},
  year={2019},
  publisher={Elsevier}
}

\newpage

\appendix

\section{Proof of theorem \ref{theo:bias_theorem}}\label{app:proof_theorem}
\begin{proof}
    Given an exact decoupling of $\mathbf{f}$ as in equation \eqref{eq:k-layer-decoupling} with a parameterized representation of the internal functions as outlined in Section \ref{sec:param_functions}, %\eqref{eq:decoupled_analytic}, 
    an equivalent decoupling of the form of equation \eqref{eq:decoupled_analytic_no_constant} can be constructed by induction on the layers. Namely, given the decoupling from equation \eqref{eq:k-layer-decoupling}, %\eqref{eq:decoupled_analytic}, 
    it holds that
    \begin{align}
        \mathbf{g}_1(\mathbf{W}_0\mathbf{x}) - \mathbf{b}_1 &= \hat{\mathbf{g}}_1(\mathbf{W}_0 \mathbf{x}) =: \mathbf{u}_1, \nonumber
    \end{align}
    where the constant terms of $\hat{\mathbf{g}}_1$ are equal to zero and 
    \begin{equation}
        \mathbf{b}_1 = \begin{bmatrix}
            c^{(1)}_{1,0} & c^{(2)}_{1,0} & \cdots & c^{(r_1)}_{1,0}
        \end{bmatrix}^\top. \nonumber
    \end{equation}
    Thus the result of the second layer, with $\hat{\mathbf{u}}_1 := \mathbf{W}_1 \; \mathbf{u}_1$ and $\hat{\mathbf{b}}_1 := \mathbf{W}_1 \; \mathbf{b}_1$, becomes
    \begin{align}
        \mathbf{g}_2(\mathbf{W}_1 \;(\mathbf{u}_1 + \mathbf{b}_1)) &= \mathbf{g}_2(\mathbf{W}_1 \; \mathbf{u}_1 + \mathbf{W}_1 \; \mathbf{b}_1) \nonumber \\
        &= \mathbf{g}_2(\hat{\mathbf{u}}_1 + \hat{\mathbf{b}}_1). \nonumber
    \end{align}
    As mentioned in Theorem \ref{theo:bias_theorem}, the input $\hat{\mathbf{u}}_1 + \hat{\mathbf{b}}_1$ belongs to an open set $\Omega_2$ such that each internal function $g^{(j)}_2$, for $j=1,2,\hdots,r_2$, is real analytic on $\Omega_2$. As a result, each internal function can be represented as a convergent Taylor series
    \begin{equation}
        \begin{split}
            g^{(j)}_2(\hat{u}_{1,j} + \hat{b}_{1,j}) -  g^{(j)}_2(\hat{b}_{1,j}) &= \sum^{\infty}_{n=1} \dfrac{g^{(j)}_2(\hat{b}_{1,j})^{(n)}}{n!} \cdot (\hat{u}_{1,j})^n \\
            &=: \hat{g}^{(j)}_2(\hat{u}_{1,j}), \label{eq:construction_g_2_hat}
        \end{split}
    \end{equation}
    where $g^{(j)}_2(\hat{b}_{1,j})$ is a constant and $\hat{g}^{(j)}_2(\hat{u}_{1,j})$ a function only dependent on $\hat{u}_{1,j}$. From equation \eqref{eq:construction_g_2_hat} it follows that
    \begin{equation}
        \begin{split}
            \hat{\mathbf{g}}_2(\hat{\mathbf{u}}_{1}) := \mathbf{g}_2(\hat{\mathbf{u}}_1 + \hat{\mathbf{b}}_1) - \mathbf{g}_2(\hat{\mathbf{b}}_{1}),
        \end{split}
    \end{equation}
    showing that the bias from the first layer can be removed 
    %(constant coefficients equal to zero) 
    and absorbed into the coefficients of the power series representation of the second layer. This constitutes the base case of the induction.
    
    For the induction step, by assuming that the bias of layer $\ell - 2$ can be moved to layer $\ell - 1$ and that $\mathbf{g}_{\ell-1}(\hat{\mathbf{u}}_{l-2} + \hat{\mathbf{b}}_{\ell-2})  = \mathbf{g}_{\ell-1}(\hat{\mathbf{b}}_{\ell-2}) + \hat{\mathbf{g}}_{\ell-1}(\hat{\mathbf{u}}_{\ell-2})$, an analogous reasoning as for the base step can be used to show that
    \begin{equation}
        \begin{split}
            \mathbf{g}_\ell(\hat{\mathbf{u}}_{\ell-1} + \hat{\mathbf{b}}_{\ell-1}) &= \mathbf{g}_{\ell}(\hat{\mathbf{b}}_{\ell-1}) + \hat{\mathbf{g}}_{\ell}(\hat{\mathbf{u}}_{\ell-1}),
        \end{split}
    \end{equation}
    for $\ell \leq L-1$ and $\mathbf{g}_L(\hat{\mathbf{u}}_{L-1} + \hat{\mathbf{b}}_{L-1}) = \mathbf{g}_L(\hat{\mathbf{b}}_{L-1}) + \hat{\mathbf{g}}_L(\hat{\mathbf{u}}_{L-1}) =: \bar{\mathbf{g}}_L(\hat{\mathbf{u}}_{L-1})$ Thus, it is possible to construct an exact decoupling of $\mathbf{f}$
    \begin{equation}
        \mathbf{f}(\mathbf{x}) = \mathbf{W}_L \; \bar{\mathbf{g}}_L(\mathbf{W}_{L-1} \; \hat{\mathbf{g}}_{L-1}(\hdots \hat{\mathbf{g}}_1(\mathbf{W}_0 \mathbf{x})\hdots)),
    \end{equation}
    where only the constant terms of the internal functions of $\bar{\mathbf{g}}_L$ may be non-zero. In addition, the construction of the internal functions, as done for $\hat{g}^{(j)}_2$ in equation \eqref{eq:construction_g_2_hat}, shows them to be convergent power series.
\end{proof}

\section{Structure of matrices in (PROJ-)CMPT-$L$ algorithm}\label{app:mat_proj}
\noindent
The structure of matrices $\mathbf{M}^{(0)}_{\mathbf{W}}$, $\mathbf{M}^{(\ell_1)}_{\mathbf{W}}$, $\mathbf{M}^{(L)}_{\mathbf{W}}$, $\mathbf{M}^{(\ell_2,s)}_{\mathbf{G}}$, $\mathbf{M}^{(L,s)}_{\mathbf{G}}$ and $\mathbf{M}^{(1,s)}_{\mathbf{G}}$, for $\ell_1=1,2,\hdots, L-1$, $\ell_2=2,3,\hdots,L-1$ and $s=1,2,\hdots,S$, are given by

\begin{equation}
    \begin{split}
        \mathbf{M}^{(0)}_{\mathbf{W}} &= \begin{bmatrix}
            \mathbf{W}_L \; \mathbf{D}^{(1)}_{L} \; \mathbf{W}_{L-1}\; \cdots \; \mathbf{W}_{1} \; \mathbf{D}^{(1)}_{1} \\
            \mathbf{W}_L \; \mathbf{D}^{(2)}_{L} \; \mathbf{W}_{L-1}\; \cdots \; \mathbf{W}_{1} \; \mathbf{D}^{(2)}_{1} \\
            \vdots \\
            \mathbf{W}_L \; \mathbf{D}^{(S)}_{L} \; \mathbf{W}_{L-1}\; \cdots \; \mathbf{W}_{1} \; \mathbf{D}^{(S)}_{1}
        \end{bmatrix}, \\
        \mathbf{M}^{(L)}_{\mathbf{W}} &= \begin{bmatrix}
            (\mathbf{D}^{(1)}_{L} \; \mathbf{W}_{L-1}\; \cdots \; \mathbf{W}_{1} \; \mathbf{D}^{(1)}_{1} \; \mathbf{W}_0)^\top \\
            (\mathbf{D}^{(2)}_{L} \; \mathbf{W}_{L-1}\; \cdots \; \mathbf{W}_{1} \; \mathbf{D}^{(2)}_{1} \; \mathbf{W}_0)^\top \\
            \vdots \\
            (\mathbf{D}^{(S)}_{L} \; \mathbf{W}_{L-1}\; \cdots \; \mathbf{W}_{1} \; \mathbf{D}^{(S)}_{1} \; \mathbf{W}_0)^\top
        \end{bmatrix}^\top, \\
        \mathbf{M}^{(\ell_1)}_{\mathbf{W}} &= \begin{bmatrix}
            (\mathbf{D}^{(1)}_{\ell_1} \; \mathbf{W}_{\ell_1-1} \cdots \mathbf{W}_{0})^\top \otimes (\mathbf{W}_L \; \mathbf{D}^{(1)}_{L} \cdots \; \mathbf{D}^{(1)}_{\ell_1+1}) \\
            (\mathbf{D}^{(2)}_{\ell_1} \; \mathbf{W}_{\ell_1-1} \cdots \mathbf{W}_{0})^\top \otimes (\mathbf{W}_L \; \mathbf{D}^{(2)}_{L} \cdots \; \mathbf{D}^{(2)}_{\ell_1+1}) \\
            \vdots \\
            (\mathbf{D}^{(S)}_{\ell_1} \; \mathbf{W}_{\ell_1-1} \cdots \mathbf{W}_{0})^\top \otimes (\mathbf{W}_L \; \mathbf{D}^{(S)}_{L} \cdots \; \mathbf{D}^{(S)}_{\ell_1+1})
        \end{bmatrix}, \\
        \mathbf{M}^{(\ell_2,s)}_{\mathbf{G}} &= \begin{bmatrix}
            (\mathbf{W}_{\ell_2-1} \; \mathbf{D}^{(s)}_{\ell_2-1} \; \cdots \; \mathbf{W}_{0})^\top \odot (\mathbf{W}_L \; \mathbf{D}^{(s)}_{L} \cdots \; \mathbf{W}_{\ell_2})
        \end{bmatrix}, \\
        \mathbf{M}^{(L,s)}_{\mathbf{G}} &= \begin{bmatrix}
            (\mathbf{W}_{L-1} \; \mathbf{D}^{(s)}_{L-1} \; \cdots \; \mathbf{W}_{0})^\top \odot \mathbf{W}_L
        \end{bmatrix}, \\
        \mathbf{M}^{(1,s)}_{\mathbf{G}} &= \begin{bmatrix}
            \mathbf{W}_{0}^\top \odot (\mathbf{W}_L \; \mathbf{D}^{(s)}_{L} \cdots \; \mathbf{W}_1)
        \end{bmatrix}, \nonumber
    \end{split}
\end{equation}
where 
\begin{align}
    &\mathbf{M}^{(0)}_{\mathbf{W}} \in \mathbb{R}^{nS \times r_1},\;\mathbf{M}^{(L)}_{\mathbf{W}} \in \mathbb{R}^{r_L \times mS}, \nonumber \\
    &\mathbf{M}^{(\ell_1)}_{\mathbf{W}} \in \mathbb{R}^{nmS \times r_{\ell_1} r_{\ell_1+1}}, \;\mathbf{M}^{(\ell_2,s)}_{\mathbf{D}} \in \mathbb{R}^{nm \times r_{\ell_2}}, \nonumber \\
    &\mathbf{M}^{(L,s)}_{\mathbf{G}} \in \mathbb{R}^{nm \times r_L},\;\mathbf{M}^{(1,s)}_{\mathbf{G}} \in \mathbb{R}^{nm \times r_1}. \nonumber
\end{align}

\section{Structure of matrices in CONSTR-PT-$L$ algorithm}\label{app:mat_constr}

\noindent In CONSTR-PT-$L$, the updates can be derived by noticing that
\begin{equation}
    \begin{split}
        \mathbf{D}^{(s)}_{\ell} = \begin{bmatrix}
            \boldsymbol{\phi}'_{\ell}(u^{(s)}_{\ell, 1}) &  & \mathbf{0} \\
             & \ddots &  \\
            \mathbf{0} &  & \boldsymbol{\phi}'_{\ell}(u^{(s)}_{{\ell, r_{\ell}}}) 
        \end{bmatrix}
        \cdot
        \begin{bmatrix}
            \mathbf{c}^{1}_{\ell} &  & \mathbf{0} \\
             & \ddots &  \\
            \mathbf{0} &  & \mathbf{c}^{r_{\ell}}_{\ell}
        \end{bmatrix} = \mathbf{X}^{(s)}_{\ell} \cdot \mathbf{C}_{\ell},
        \label{eq:incorp_struct_D}
    \end{split}
\end{equation}
where $\mathbf{X}^{(s)}_{\ell} \in \mathbb{R}^{r_{\ell} \times r_{\ell} (d_{\ell} + 1)}$, $\mathbf{C}_{\ell} \in \mathbb{R}^{r_{\ell}(d_{\ell} + 1) \times r_{\ell}}$ and
\begin{equation}
    \boldsymbol{\phi}'_{\ell}(u^{(s)}_{\ell, j}) = 
    \begin{bmatrix} 0 \\ \phi'_{\ell,1}(u^{(s)}_{\ell, j}) \\ \vdots \\ \phi'_{\ell,d_{\ell}}(u^{(s)}_{\ell, j})\end{bmatrix}^\top,        \; \mathbf{c}^j_{\ell} = \begin{bmatrix}
    c^{(j)}_{\ell,0} \\ \vdots \\ c^{(j)}_{\ell,d_{\ell}}
    \end{bmatrix}, \nonumber
\end{equation}
for $\boldsymbol{\phi}'_{\ell}(u^{(s)}_{\ell,j}) \in \mathbb{R}^{1 \times (d_{\ell} + 1)}$ and $\mathbf{c}^j_{\ell} \in \mathbb{R}^{(d_{\ell} + 1) \times 1}$.

Now, by incorporating this structure, the frontal slices of $\mathcal{J}$ become
\begin{align}
    \mathcal{J}_{:,:,s} &= \mathbf{W}_L \cdot \mathbf{D}^{(1)}_{L} \cdot \mathbf{W}_{L-1} \cdots \mathbf{D}^{(1)}_{1} \cdot \mathbf{W}_0 \nonumber \\
    &= \mathbf{W}_L \cdot (\mathbf{X}^{(s)}_{L} \cdot \mathbf{C}_L) \cdot \mathbf{W}_{L-1} \cdots (\mathbf{X}^{(s)}_{1} \cdot \mathbf{C}_1) \cdot \mathbf{W}_{0}, \label{eq:struct_J_constr_incorp}  
\end{align}
\begin{figure*}
    \centering
    \includegraphics[width=1\textwidth]{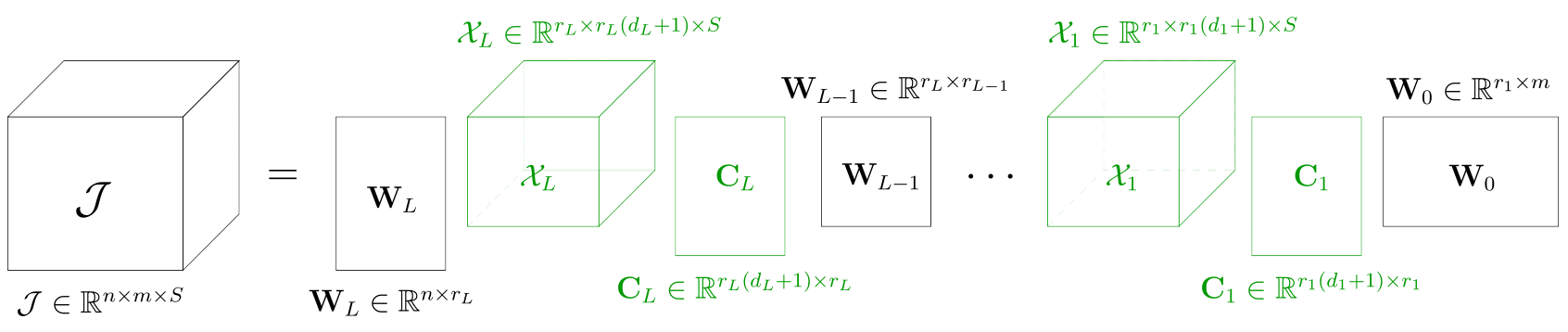}
    \caption{Decomposition of $\mathcal{J}$ according to equation \eqref{eq:struct_J_constr_incorp}.}
    \label{fig:CONSTR-PT-k}
\end{figure*}
for $s=1,2,\hdots,S$. Figure \ref{fig:CONSTR-PT-k} shows the decomposition of $\mathcal{J}$ with the incorporated structure from equation \eqref{eq:struct_J_constr_incorp}. Next, the structure from equation \eqref{eq:struct_J_constr_incorp} can be used to derive ALS updates for each of the $\mathbf{C}_\ell$ matrices, for $\ell=1,2,\hdots,L$. Namely,
\begin{align}
        \mathbf{F}^{(s)}_{Le, \ell} &= \mathbf{W}_L \cdot (\mathbf{X}^{(s)}_{L} \cdot \mathbf{C}_L) \cdot \mathbf{W}_{L-1} \cdots \mathbf{X}^{(s)}_{\ell}  \nonumber \\
        \mathbf{F}^{(s)}_{Ri, \ell} &= \mathbf{W}_{\ell-1} \cdot (\mathbf{X}^{(s)}_{\ell-1} \cdot \mathbf{C}_{\ell-1}) \cdot \mathbf{W}_{\ell-2} \cdots \mathbf{W}_{0} \nonumber \\
        \mathcal{J}_{:,:,s} &= \mathbf{F}^{(s)}_{Le, \ell} \cdot \mathbf{C}_{\ell} \cdot \mathbf{F}^{(s)}_{Ri, \ell}, \label{eq:isolate_C}
\end{align}
such that using equation \eqref{eq:isolate_C} it holds that
\begin{gather}
    \text{vec}(\mathcal{J}_{:,:,s}) = (\mathbf{F}^{(s)T}_{Ri, \ell} \otimes \mathbf{F}^{(s)}_{Le, \ell}) \cdot \text{vec}(\mathbf{C}_{\ell}) \nonumber \\
    \vspace{10px}
    \Downarrow 
    \vspace{10px}
    \nonumber \\
    \text{vec}(\mathcal{J}) = \begin{bmatrix}
        \mathbf{F}^{(1)T}_{Ri, \ell} \otimes \mathbf{F}^{(1)}_{Le, \ell} \\
        \mathbf{F}^{(2)T}_{Ri, \ell} \otimes \mathbf{F}^{(2)}_{Le, \ell} \\
        \vdots \\
        \mathbf{F}^{(S)T}_{Ri, \ell} \otimes \mathbf{F}^{(S)}_{Le, \ell}
    \end{bmatrix} \cdot \text{vec}(\mathbf{C}_{\ell}) = \mathbf{M}^{(\ell)}_{\mathbf{C}} \cdot \text{vec}(\mathbf{C}_{\ell}).  \nonumber
\end{gather}
Next, removing the zero rows from $\text{vec}(\mathbf{C}_{\ell})$ as well as the corresponding columns from $\mathbf{M}^{(i)}_{\mathbf{C}}$ (denoted as $(\mathbf{M}^{(\ell)}_{\mathbf{C}})_0)$ yields the update formula
\begin{align}
            \text{vec}(\mathcal{J}) &= (\mathbf{M}^{(\ell)}_{\mathbf{C}})_0 \cdot \text{vec}(\mathbf{C}_{\ell})_0  \nonumber\\
            &= (\mathbf{M}^{(\ell)}_{\mathbf{C}})_0 \cdot \mathbf{c}_{\ell}.\label{eq:update_C_i} \nonumber
\end{align}

\section{Systems for synthetic examples} \label{app:synthetic}

The systems $\mathbf{f}_1(\mathbf{x})$, $\mathbf{f}_2(\mathbf{x})$, $\mathbf{f}_3(\mathbf{x})$, for the synthetic examples are defined as follows,

\begin{gather}
    \mathbf{f}_{1}(\mathbf{x}): \mathbb{R}^{2} \rightarrow \mathbb{R}^2: \mathbf{x} \mapsto \mathbf{W}_2\mathbf{g}_2(\mathbf{W}_1\mathbf{g}_1(\mathbf{W}_0\mathbf{x})), \nonumber \\
    \mathbf{W}_2 = \begin{bmatrix}
        1.61 & -1.9 \\
        -0.03 & 0.11
    \end{bmatrix},\; \mathbf{W}_1 = \begin{bmatrix}
        0.87 & -0.99 \\
        -1.42 & 0.9
    \end{bmatrix},\; \mathbf{W}_0 = \begin{bmatrix}
         1.72 & -0.73 \\
        -1.26 & -1.18
    \end{bmatrix}  \nonumber \\
    \mathbf{g}_1(\mathbf{u}) = \begin{bmatrix}
        1.91 u^5_1 &+& 1.37u^4_1 &+& 2.37u^3_1 &-& 2.69u^2_1 &+& 0.58u_1 \\
        -1.45 u^5_2 &-& 1.69u^4_2 &-& 2.42u^3_2 &+& 1.86u^2_2 &
    \end{bmatrix}, \nonumber \\
    \mathbf{g}_2(\mathbf{v}) = \begin{bmatrix}
        1.26v^2_1 &-& 0.24 v_1 &-& 0.19 \\
        -1.99v^2_2 &+& 0.19 v_2 &-& 1.93
    \end{bmatrix}, \nonumber
\end{gather}
\begin{gather}
    \mathbf{f}_{2}(\mathbf{x}): \mathbb{R}^{3} \rightarrow \mathbb{R}^{3}: \mathbf{x} \mapsto \mathbf{W}_2\mathbf{g}_2(\mathbf{W}_1\mathbf{g}_1(\mathbf{W}_0\mathbf{x})), \nonumber \\
    \mathbf{W}_2 = \begin{bmatrix}
        -0.59 & 0.86 \\
        0.02 & -1.1 \\
        -1.02 & 1.17
    \end{bmatrix},\; \mathbf{W}_1 = \begin{bmatrix}
        0.21 & -0.94 \\
        -1.12 & 0.56
    \end{bmatrix},\; \nonumber \\
    \mathbf{W}_0 = \begin{bmatrix}
         1.08 & 1.71 & 0.44 \\
        -1.4 & -0.04 & -0.49
    \end{bmatrix}  \nonumber \\
    \mathbf{g}_1(\mathbf{u}) = \begin{bmatrix}
        2.67u^3_1 &+& 2.49u^2_1 &-& 0.03u_1 \\
        1.33u^3_2 &-& 1.49u^2_2 &+& 0.2u_2
    \end{bmatrix}, \nonumber \\
    \mathbf{g}_2(\mathbf{v}) = \begin{bmatrix}
        -0.88v^3_1 &-& 1.64v^2_1 &-& 0.01 v_1 &-& 0.8 \\
        1.69v^3_2 &-& 1.61v^2_2 &-& 1.12 v_2 &+& 0.91
    \end{bmatrix}, \nonumber
\end{gather}
\begin{gather}
    \mathbf{f}_{3}(\mathbf{x}): \mathbb{R}^{4} \rightarrow \mathbb{R}^{3}: \mathbf{x} \mapsto \mathbf{W}_2\mathbf{g}_2(\mathbf{W}_1\mathbf{g}_1(\mathbf{W}_0\mathbf{x})), \nonumber \\
    \mathbf{W}_2 = \begin{bmatrix}
        1.05 & -1.11 \\
        -0.95 & -0.17 \\
        -1. & 0.27
    \end{bmatrix},\; \mathbf{W}_1 = \begin{bmatrix}
        1.49 & -1.05 & -1 \\
        0.78 & -1.29 & 0.62
    \end{bmatrix},\nonumber \\ \mathbf{W}_0 = \begin{bmatrix}
        -1.43 &  0.06 &  0.76 & 1.43 \\
        0.59 & 0.33 & 0.84 & -0.99 \\
        1.6 & -0.23 & -1.92 & 1.84
    \end{bmatrix}  \nonumber \\
    \mathbf{g}_1(\mathbf{u}) = \begin{bmatrix}
        -0.41u^3_1 &-& 0.73u^2_1 &+& 2.08u_1 \\
        -2.76u^3_2 &-& 0.77u^2_2 &+& 2u_2 \\
        1.35u^3_3 &-& 0.29u^2_3 &+& 0.33u_3
    \end{bmatrix}, \nonumber \\
    \mathbf{g}_2(\mathbf{v}) = \begin{bmatrix}
        0.97v^4_1 & + & 0.38v^3_1 & - & 0.18v^2_1 & + & 2.04v_1 & - & 0.73  \\
        -0.71v^4_2 &+& 1.4v^3_2 &-& 1.67v^2_2 &+& 0.74v_2 & - & 0.23
    \end{bmatrix}, \nonumber
\end{gather}

The elements in the weight matrices $\mathbf{W}_2$ and $\mathbf{W}_0$ have been generated with elements sampled from the uniform distribution $\mathcal{U}(-2,2)$ and the elements from $\mathbf{W}_1$ were sampled from $\mathcal{U}(-1.5,1.5)$, under the restriction that the maximal collinearity factor $C$ of these factor matrices,
\begin{equation}
    C = \text{max}\left[\left\{\dfrac{\mathbf{W}^{:,i} \mathbf{W}^{:,j}}{\lVert \mathbf{W}^{:,i} \rVert \lVert \mathbf{W}^{:,j} \rVert} \;\; \Big| \;\; i\neq j\right\}\right], \nonumber
\end{equation}
is lower than $0.5$. The polynomial coefficients are sampled from $\mathcal{U}(-3,3)$.

\end{document}